\documentclass[pra,twocolumn, preprintnumbers,amsmath,amssymb, nofootinbib]{revtex4}

\usepackage{graphicx, color, graphpap}
\usepackage{amsmath}
\usepackage{enumitem}
\usepackage{amssymb}
\usepackage{amsthm}
\usepackage{pstricks}
\usepackage{float}
\usepackage{multirow}
\usepackage{hyperref}
\usepackage{overpic}
\usepackage[T1]{fontenc}
\long\def\ca#1\cb{} 

\newcommand{\ket}[1]{|#1\rangle}               
\newcommand{\bra}[1]{\langle #1|}              
\newcommand{\dya}[1]{\ket{#1}\!\bra{#1}}
\newcommand{\dyad}[2]{\ket{#1}\!\bra{#2}}        
\newcommand{\ip}[2]{\langle #1|#2\rangle}      

\newcommand{\CC}{\mathcal{C}}
\newcommand{\DC}{\mathcal{D}}
\newcommand{\EC}{\mathcal{E}}
\newcommand{\FC}{\mathcal{F}}
\newcommand{\GC}{\mathcal{G}}
\newcommand{\HC}{\mathcal{H}}
\newcommand{\IC}{\mathcal{I}}

\newcommand{\LC}{\mathcal{L}}
\newcommand{\MC}{\mathcal{M}}

\newcommand{\PC}{\mathcal{P}}

\newcommand{\SC}{\mathcal{S}}

\newcommand{\ZC}{\mathcal{Z}}

\newcommand{\Tr}{{\rm Tr}}

\newcommand{\pass}{\text{pass}}
\newcommand{\ave}[1]{\langle #1\rangle}               
\renewcommand{\geq}{\geqslant}
\renewcommand{\leq}{\leqslant}

\newcommand{\mted}[3]{\langle#1|#2|#3\rangle }

\newcommand{\ot}{\otimes}
\newcommand{\ad}{^\dagger}

\newcommand*{\id}{\openone}

\newcommand{\rholh}{\hat{\rho}}

\newcommand{\rhot}{\tilde{\rho}}

\newcommand{\al}{\alpha }

\newcommand{\gm}{\gamma }
\newcommand{\Gm}{\Gamma }

\newcommand{\Th}{\Theta }

\newcommand{\lm}{\lambda }
\newcommand{\lmv}{\vec{\lambda} }
\newcommand{\gmv}{\vec{\gamma} }
\newcommand{\Gmv}{\vec{\Gamma} }

\newcommand{\sg}{\sigma }

\newcommand{\om}{\omega }
\newcommand{\Om}{\Omega }

\newcommand{\reg}{M}

\newtheoremstyle{example}{\topsep}{\topsep}%
{}
{\parindent}
{\itshape}
{:}
{   }
{\thmname{#1}\thmnumber{ #2}}
\theoremstyle{example}

\newtheorem{theorem}{Theorem}
\newtheorem{lemma}[theorem]{Lemma}

\newtheorem{proposition}[theorem]{Proposition}

\theoremstyle{definition}

\newcommand{\gt}{(59)}
\newcommand{\gtapplied}{(60)}
\newcommand{\sigmaoptimal}{(63)}
\newcommand{\alphacoherence}{(43)}
\newcommand{\coherencesigma}{(50)}
\newcommand{\strongduality}{(48)}
\newcommand{\devwin}{(1)}
\newcommand{\constraints}{(5)}

\begin{document}

\title{Numerical approach for unstructured quantum key distribution}

\author{Patrick J. Coles}
\email{pcoles@uwaterloo.ca}
\affiliation{Institute for Quantum Computing and Department of Physics and Astronomy, University of Waterloo, N2L3G1 Waterloo, Ontario, Canada}

\author{Eric M. Metodiev}
\affiliation{Institute for Quantum Computing and Department of Physics and Astronomy, University of Waterloo, N2L3G1 Waterloo, Ontario, Canada}

\author{Norbert L\"utkenhaus}
\affiliation{Institute for Quantum Computing and Department of Physics and Astronomy, University of Waterloo, N2L3G1 Waterloo, Ontario, Canada}

\begin{abstract}
Quantum key distribution (QKD) allows for communication with security guaranteed by quantum theory. The main theoretical problem in QKD is to calculate the secret key rate for a given protocol. Analytical formulas are known for protocols with symmetries, since symmetry simplifies the analysis. However, experimental imperfections break symmetries, hence the effect of imperfections on key rates is difficult to estimate. Furthermore, it is an interesting question whether (intentionally) asymmetric protocols could outperform symmetric ones. Here, we develop a robust numerical approach for calculating the key rate for arbitrary discrete-variable QKD protocols. Ultimately this will allow researchers to study ``unstructured'' protocols, that is, those that lack symmetry. Our approach relies on transforming the key rate calculation to the dual optimization problem, which dramatically reduces the number of parameters and hence the calculation time. We illustrate our method by investigating some unstructured protocols for which the key rate was previously unknown.
\end{abstract}

\break
\newpage
\newpage
\maketitle

Quantum key distribution (QKD) will play an important role in quantum-safe cryptography, i.e., cryptography that addresses the emerging threat of quantum computers \cite{Campagna2015}. Since its original proposal~\cite{Wiesner1983, Bennett1984}, QKD has developed dramatically over the past three decades~\cite{Scarani2009, Lo2014}, both in theory and implementation. Indeed, QKD is now a commercial technology, with the prospect of global QKD networks on the horizon~\cite{Sasaki2011, Wang2014}.

The main theoretical problem in QKD is to calculate how much secret key can be distributed by a given protocol. A crucial practical issue is that the QKD protocols that are easiest to implement with existing optical technology do not necessarily coincide with the protocols that are easiest to analyze theoretically \cite{Scarani2009}. Currently, calculating the secret key output of a protocol is typically extremely technical, and hence only performed by skilled experts. Furthermore, each new protocol idea requires a new calculation, tailored to that protocol. Ultimately the technical nature of these calculations combined with the lack of universal tools limits the pace at which new QKD protocols can be discovered and analyzed. Here, we address this problem by developing a robust, user-friendly framework for calculating the secret key output, with the hope of bringing such calculations ``to the masses''.

The secret key output is typically quantified by the key rate, which refers to the number of bits of secret key established divided by the number of distributed quantum systems. Operationally this corresponds to the question of how much privacy amplification Alice and Bob must apply to transform their raw key into the final secure key. Analytical simplifications of the key rate calculation can be made for some special protocols that have a high degree of symmetry \cite{Ferenczi2012}. Examples of such symmetric protocols, where the signal states have a group-theoretic structure, include the BB84 \cite{Bennett1984} and six-state protocols \cite{Bruss2002}. Indeed the key rate is known for these protocols. However, in practice, lack of symmetry is the rule rather than the exception. That is, even if experimentalists try to implement a symmetric protocol, experimental imperfections tend to break symmetries \cite{Gottesman2004}. Furthermore, it is sometimes desirable due to optical hardware issues to implement asymmetric protocols, e.g., as in Ref.~\cite{Fung2006}.

We refer to general QKD protocols involving signal states or measurement choices that lack symmetry as ``unstructured'' protocols. Some recent work has made progress in bounding the key rate for special kinds of unstructured protocols, such as four-state protocols in Ref.~\cite{Maroy2010, Woodhead2013} and qubit protocols in Ref.~\cite{Tamaki2014}. Still, there is no general method for computing tight bounds on the key rate for arbitrary unstructured protocols. Yet, these are the protocols that are most relevant to experimental implementations.

This motivates our present work, in which we develop an efficient, numerical approach to calculating key rates. Our ultimate aim is to develop a computer program, where Alice and Bob input a description of their protocol (e.g., their signal states, measurement devices, sifting procedure, and key map) and their experimental observations, and the computer outputs the key rate for their protocol. This program would allow for any protocol, including those that lack structure.

At the technical level, the key rate problem is an optimization problem, since one must minimize the well-known entropic formula for the key rate \cite{Devetak2005} over all states $\rho_{AB}$ that satisfy Alice's and Bob's experimental data. The main challenge here is that this optimization problem is unreliable and inefficient. In this work, we give a novel insight that transforming to the dual problem (e.g., see \cite{Boyd2010}) resolves these issues, hence paving the way for automated key rate calculations.

Specifically, the unreliable (or unphysical) aspect of the primal problem is that it is a minimization, hence the output will in general be an upper bound on the key rate. But one is typically more interested in reliable lower bounds, i.e., physically achievable key rates. Transforming to the dual problem allows one to formulate the problem as a maximization, and hence approach the key rate from below. Therefore, every number outputted from our computer program represents an achievable asymptotic key rate, even if the computer did not reach the global maximum.

The inefficient aspect of the primal problem is that the number of parameters grows as $d_A^2d_B^2 $ for a state $\rho_{AB}$ with $d_A=\dim(\HC_A)$ and $d_B=\dim(\HC_B)$. For example, if $d_A = d_B = 10$, the number of parameters that one would have to optimize over is 10000. In contrast, in the dual problem, the number of parameters is equal to the number of experimental constraints that Alice and Bob choose to impose.  For example, in the generalization of the BB84 protocol to arbitrary dimensions \cite{Sheridan2010,Mafu2013}, Alice and Bob typically consider two constraints, their error rates in the two mutually-unbiased bases (MUBs).  So, for this protocol, we have reduced the number of parameters to something that is constant in dimension. We therefore believe that our approach (of solving the dual problem) is ideally suited to efficiently calculate key rates in high dimensions.

We have written a MATLAB program to implement our key rate calculations. To illustrate the validity of our program, we show (see Fig.~\ref{fgrBB84sixstate}) that it exactly reproduces the known theoretical dependence of the key rate on error rate, for both the BB84 and six-state protocols. 

But ultimately the strength of our approach is its ability to handle unstructured protocols. We demonstrate this by investigating some unstructured protocols for which the key rates were not previously known. For example, we study a general class of protocols where Alice and Bob measure $n$ MUBs, with $2\leq n \leq d+1$, in dimension $d$. Also, we investigate the B92 protocol \cite{Bennett1992}, which involves two signal states whose inner product is arbitrary. Our key rates are higher than known analytical lower bounds \cite{Tamaki2004, Renner2005a} for B92. Finally, we argue that our approach typically gives dramatically higher key rates than those obtained from an analytical approach based on the entropic uncertainty relation \cite{Berta2010,Tomamichel2012a}.

We focus on asymptotic key rates in this work. Nevertheless, the optimization problem that we solve is also at the heart of finite-key analysis, e.g., see Refs.~\cite{Scarani2008, Sano2010}. We therefore hope to extend our approach to the finite-key scenario in future efforts. We remark that finite-size effects generally reduce the key rate below its asymptotic value.

In what follows we first present our main result: a reformulation of the key rate optimization problem in such a way that it is easily computable. We then outline our general framework for treating a broad range of protocols. Finally we illustrate our approach with various examples.

\begin{figure}
\begin{center}
\includegraphics[width=3.3in]{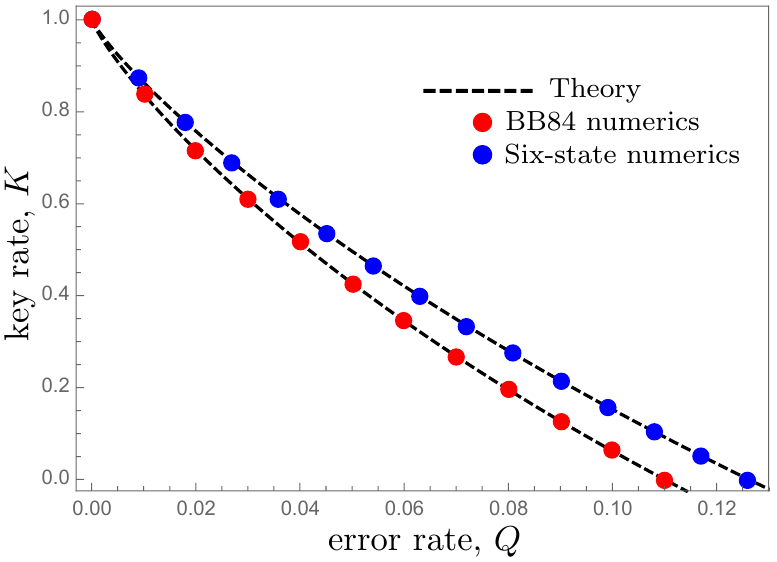}
\caption{Key rate for two well-known QKD protocols. Here we compare our numerics (from Theorem~\ref{thm1}) with the theoretical curves. The results of our numerical optimization for the BB84 and six-state protocols are respectively shown as red and blue dots. The known theoretical curves for these protocols are also shown as black dashed lines. The dots should be viewed as reliable lower bounds on the key rate, but in this case they happen to be perfectly tight, coinciding with the theoretical curves.}
\label{fgrBB84sixstate}
\end{center}
\end{figure}

\section*{\uppercase{Results}}

\noindent\textbf{Setup of the problem.} Consider a general entanglement-based (EB) QKD protocol involving finite-dimensional quantum systems $A$ and $B$ that are respectively received by Alice and Bob. Note that prepare-and-measure QKD protocols can be recast as EB protocols, as discussed below. For simplicity of presentation, we consider protocols where Alice's raw key is derived from a measurement on her system, possibly after some post-selection corresponding to a public announcement with a binary outcome, ``pass'' or ``fail''. However, our approach can easily be extended to more general protocols.

Let $Z_A$ ($Z_B$) denote the measurement that Alice (Bob) performs on system $A$ ($B$) in order to derive the raw key.  Suppose they use one-way direct reconciliation for the classical post-processing and that their error correction is optimal (i.e., leaks out the minimum number of bits), then the asymptotic key rate is given by the Devetak-Winter formula \cite{Devetak2005}:
\begin{align}
\label{eqnDevWin1}
K &=  H(Z_A | E) - H(Z_A | Z_B) \,.
\end{align}
In \eqref{eqnDevWin1}, $H(X|Y) := H(\rho_{XY}) - H(\rho_{Y})$ is the conditional von Neumann entropy, with $H(\sigma) := -\Tr (\sigma\log_2 \sigma)$, and
\begin{align}
\label{eqncqccstates}
\rho_{Z_A Z_B} &= \sum_{j,k} \Tr [(Z_A^j\ot Z_B^k) \rho_{AB}] \dya{j} \ot \dya{k} ,\\ 
\rho_{Z_A E} &= \sum_j \dya{j} \ot \Tr_A[(Z_A^j\ot \id) \rho_{AE}] .
\end{align}
Here, $\rho_{ABE}$ is the tripartite density operator shared by Alice, Bob, and Eve (and it may be the state after some post-selection procedure, see our general framework below). Also, $ \{Z_A^j\}$ and $\{Z_B^k\}$ are the sets of positive operator valued measure (POVM) elements associated with Alice's and Bob's key-generating measurements. In what follows we refer to $ \{Z_A^j\}$ as the key-map POVM.

In the previous paragraph and in what follows, we assume that the state shared by Alice, Bob, and Eve has an i.i.d.\ (independent, identically distributed) structure, and hence it makes sense to discuss the state $\rho_{ABE}$ associated with a single round of quantum communication. To avoid confusion, we emphasize that our approach is ``unstructured'' in the sense of lacking structure for a given round of quantum communication, but we do impose the i.i.d.\ structure that relates one round to the other rounds. This i.i.d.\ structure corresponds to Eve doing a so-called collective attack. However, the security of our derived asymptotic key rate also holds against the most general attacks (coherent attacks) if one imposes that the protocol involves a random permutation of the rounds (a symmetrization step) such that the de Finetti theorem \cite{Renner2005, Renner2007} or the postselection technique \cite{Christandl2009} applies.

Typically Alice's and Bob's shared density operator $\rho_{AB}$ is unknown to them. A standard part of QKD protocols is for Alice and Bob to gather data through local measurements, and in a procedure known as parameter estimation, they use this data to constrain the form of $\rho_{AB}$. The measurements used for this purpose can, in general, be described by bounded Hermitian operators $\Gamma_i $, with the set of such operators denoted by $\vec{\Gamma} = \{  \Gamma_i  \}$.

From their data, Alice and Bob determine the average value of each of these measurements: 
\begin{align}
\label{eqnconstraints}
\vec{\gamma} = \{\gamma_i\}, \quad \text{with  } \gamma_i:= \ave{\Gamma_i} = \Tr(\rho_{AB}\Gamma_i) \,,
\end{align}
and this gives a set of experimental constraints:
\begin{align}
\label{eqnconstraintset}
C = \{ \Tr(\rho_{AB}\Gamma_i ) = \gamma_i\}.
\end{align}
We denote the set of density operators that are consistent with these constraints as:
\begin{align}
\label{eqnconstrainedstateset}
\CC = \{ \rho_{AB} \in \PC_{AB} : C \text{ holds}\}
\end{align}
where $\PC_{AB}$ denotes the set of positive semidefinite operators on $\HC_{AB}$, and an additional constraint $\ave{\id} = 1$ is assumed to be added to the set $C$ to enforce normalization.

Because Alice and Bob typically do not perform full tomography on the state, $\CC$ includes many density operators, and hence the term $H(Z_A|E)$ in \eqref{eqnDevWin1} is unknown. To evaluate the key rate, Alice and Bob must consider the most pessimistic of scenarios where $H(Z_A|E)$ takes on its smallest possible value that is consistent with their data. This is a constrained optimization problem, given by
\begin{align}
\label{eqnprimalproblem5}
K &= \min_{\rho_{AB}\in \CC}\left[ H(Z_A | E) - H(Z_A | Z_B) \right]
\end{align}
where Eve's system $E$ can be assumed to purify $\rho_{AB}$ since it gives Eve the most information. Here the number of parameters in the optimization is $(d_Ad_B)^2$, corresponding to the number of parameters in a positive semidefinite operator on $\HC_{AB}$. We refer to \eqref{eqnprimalproblem5} as the primal problem.

\bigskip

\noindent\textbf{Main result.} Our main result is a reformulation of the optimization problem in \eqref{eqnprimalproblem5}.

\begin{theorem}
\label{thm1}
The solution of the minimization problem in \eqref{eqnprimalproblem5} is lower bounded by the following maximization problem:
\begin{align}
\label{eqnmainresult1}
K & \geq  \frac{\Th}{\ln 2} -H(Z_A|Z_B)
\end{align}
where 
\begin{align}
\label{eqnmainresult2}
\Th &:=\max_{\lmv} \left(-\bigg\|\sum_j Z_{A}^j  R (\lmv )   Z_{A}^j\bigg\| - \lmv \cdot \vec{\gamma} \right),
\end{align}
and 
\begin{align}
\label{eqnmainresult3}
 R(\lmv ) &:= \exp \left(-\id - \lmv \cdot \vec{\Gamma} \right).
\end{align}
In \eqref{eqnmainresult2}, the optimization is over all vectors $\lmv = \{\lambda_i\}$, where the $\lambda_i$ are arbitrary real numbers and the cardinality of $\lmv$ is equal to that of $\vec{\Gamma}$. Also, $\|M \|$ denotes the supremum norm of $M$, which is the maximum eigenvalue of $M$ when $M$ is positive semidefinite, as in \eqref{eqnmainresult2}.
\end{theorem}
\bigskip

The proof of Theorem~\ref{thm1} is given in the Methods section. It relies on the duality of convex optimization problems, as well as some entropic identities that allow us to simplify the dual problem.  Note that the term $H(Z_A|Z_B)$ in \eqref{eqnmainresult1} is pulled outside of the optimization since Alice and Bob can compute it directly from their data.

The cardinalities of the sets $\lmv$ and $\vec{\Gamma}$ are the same. This means that the number of parameters $\lambda_i$ that one must optimize over, to solve \eqref{eqnmainresult2}, is equal to the number of experimental constraints that Alice and Bob have. (More precisely this is the number of independent constraints, since one can eliminate constraints that carry redundant information). This has the potential to be significantly less than the number of parameters in the primal problem. Indeed we demonstrate below that \eqref{eqnmainresult2} can be easily solved using MATLAB on a personal computer for a variety of interesting QKD protocols.

\bigskip

\noindent\textbf{Formulating constraints.} For a given protocol, how does one decide which constraints to include in the set $C$? Consider the following remarks. First, adding in more constraints will never decrease the key rate obtained from our optimization. This follows since adding a new constraint gives an additional $\lambda_i$ to maximize over, while setting this new $\lambda_i$ to zero recovers the old problem. Second, coarse-graining constraints, i.e., merging two constraints $\ave{\Gm_i} = \gm_i$ and $\ave{\Gm_j} = \gm_j$ into a single constraint $\ave{\Gm_i +\Gm_j} = \gm_i+\gm_j$, will never increase the key rate obtained from our optimization. This follows since merging two constraints means that two $\lambda_i$'s are merged into a single $\lambda_i$, thus restricting the optimization. Hence, to obtain the highest key rates, one should input all of one's refined knowledge that is available into our optimization. On the other hand, coarse-graining reduces the number of constraints and thus may help to simplify the optimization problem, possibly at the cost of a reduced key rate.

One's refined knowledge is captured as follows. In a general EB protocol, Alice measures a POVM (whose elements may be non-commuting, e.g., if she randomly measures one of two MUBs), which we denote as $\Gamma_A = \{ \Gamma_{A,i} \}$. Likewise Bob's corresponding POVM is $\Gamma_B = \{ \Gamma_{B,i} \}$. Hence, through public discussion, Alice and Bob obtain knowledge of expectation values of the form
\begin{align}
\label{eqnrefinedknowledge}
\Tr \left[\rho_{AB} (\Gamma_{A,i} \ot \Gamma_{B,j}) \right] = \gamma_{ij}\,,\quad \text{for each } i,j\,.
\end{align}
These constraints form the set $C$ in \eqref{eqnconstraintset}. We remark that it is common in the QKD field to express correlations in terms of average error rates rather than in terms of the joint probability distribution in \eqref{eqnrefinedknowledge}. This is an example of the coarse-graining that we mentioned above. For simplicity of presentation, we will do this sort of coarse-graining for some protocols that we investigate below, although \eqref{eqnrefinedknowledge} represents our general framework for constructing $C$.

\bigskip

\noindent\textbf{Framework for prepare-and-measure.} While our approach is stated in the EB scenario, let us note how it applies to prepare-and-measure (PM) protocols. Consider a PM protocol involving a set of $N$ signal states $\{\ket{\phi_j}\}$, which Alice sends with probabilities $\{p_j\}$. It is well-known that PM protocols can be recast as EB protocols using the source-replacement scheme (see, e.g., \cite{Bennett1992a, Scarani2009, Ferenczi2012}). Namely, one forms the entangled state:
\begin{align}
\label{eqnsourcereplace}
\ket{\psi_{AA'}}= \sum_j \sqrt{p_j}\ket{j}\ket{\phi_j}.
\end{align}
One imagines that Alice keeps system $A$, while system $A'$ is sent over an insecure quantum channel $\EC$ to Bob, resulting in
\begin{align}
\label{eqnsourcereplace2}
\rho_{AB} = (\IC\ot \EC)(\dya{{\psi}_{AA'}}).
\end{align}
The numerical optimization approach described above can then be applied to the state $\rho_{AB}$ in \eqref{eqnsourcereplace2}. However, in addition to the constraints obtained from Alice's and Bob's measurement results, we must add in further constraints to account for the special form of $\rho_{AB}$. In particular, note that the partial trace over $B$ gives
\begin{align}
\label{eqnsourcereplace3}
\rho_{A} = \sum_{j,k} \sqrt{p_j p_k} \ip{\phi_k}{\phi_j}\dyad{j}{k}.
\end{align}
The form of $\rho_{A}$, which is closely related to Gram matrix, depends on the inner products between the signal states, which (we assume) Alice knows. Suppose $\{\Om_i\}$ is a set of tomographically complete observables on system $A$, then one can add in the calculated expectation values $\{\om_i\}$ of these observables into the set of constraints. That is, add
\begin{align}
\label{eqnsourcereplace4}
\ave{\Om_i \ot \id} = \om_i, \quad \text{for each }i
\end{align}
to the set $C$ in \eqref{eqnconstraintset}. This will capture Alice's knowledge of her reduced density operator. 

\bigskip

\noindent\textbf{Framework for decoy states.} In decoy-state QKD \cite{Lo2005}, which aims to combat photon-number splitting attacks, Alice prepares coherent states of various intensities and then randomizes their phases before sending them to Bob. Our framework can handle decoy states simply by allowing for additional signal states to be added to the set $\{\ket{\phi_j}\}$ in \eqref{eqnsourcereplace}. For example, to treat decoy protocols with partial phase randomization \cite{Cao2015}, one can consider signal states that are bipartite (on the signal mode $S$ and the reference mode $R$) of the form
\begin{align}
\label{eqndecoy1}
\ket{\phi_{jkl }} = \ket{\alpha_j e^{i (\theta_k +\phi_l)}}_S \ot \ket{\alpha_j e^{i \theta_k}}_R
\end{align}
where $\alpha_j$ is the amplitude of the coherent state associated with the $j$-th intensity setting, $\theta_k$ is the $k$-th phase used in phase randomization, and $\phi_l$ is the phase Alice uses to encode her information (e.g., for generating key). Decoy protocols with complete phase randomization are also treatable in our framework, namely, by adding in a signal state for each photon-number basis state (up to a cut-off), and treating multi-photon signals as orthogonal states (so-called ``tagged states'') since Eve can perfectly distinguish them.

\bigskip

\noindent\textbf{Framework for MDI QKD.} A special kind of PM protocol is measurement-device-independent (MDI) QKD \cite{Lo2012}. In MDI QKD, Alice prepares states $\{\ket{\phi_j}\}$ with probabilities $\{p_j\}$ and sends them to Charlie, and Bob does the same procedure as Alice (see Fig.~\ref{fgrMDI}). Charlie typically does a Bell-basis measurement, however the security proof does not assume this particular form of measurement. Charlie announces the outcome of his measurement, which we denote by the classical register $\reg$. Our framework for treating MDI QKD considers the tripartite state $\rho_{AB\reg}$, where $A$ and $B$ respectively are Alice's and Bob's systems in the source-replacement scheme, playing the same role as system $A$ in \eqref{eqnsourcereplace} (see Supplementary Note~1 for elaboration). For our numerics, we impose the constraint that the marginal $\rho_{AB} = \rho_A\ot\rho_B$ is fixed (since Eve cannot access $A$ and $B$), with $\rho_A$ and $\rho_B$ given by the form in \eqref{eqnsourcereplace3}. We enforce this constraint using the same approach as used in \eqref{eqnsourcereplace4} to fix $\rho_A$ for PM protocols. The only other constraints we impose are the usual correlation constraints, i.e., a description of the joint probability distribution for the standard bases on $A$, $B$, and $\reg$, of the form
\begin{align}
\label{eqnmdi}
\Tr \left[\rho_{AB\reg } (\dya{j}\ot \dya{k} \ot \dya{m}  ) \right] = \gamma_{jkm}\,.
\end{align}

\bigskip

\noindent\textbf{Framework for post-selection and announcements.} In general, a QKD protocol may involve post-selection. As an example, if Alice sends photons to Bob over a lossy channel, then they may post-select on rounds in which Bob detects a photon. As noted above, for simplicity we consider protocols where the post-selection involves a binary announcement, and Alice and Bob keep (discard) rounds when ``pass'' (``fail'') is announced. Let $\GC$ be the completely-positive (CP) linear map corresponding to the post-selection. The action of $\GC(\cdot) = G(\cdot)G\ad$ is given by a single Kraus operator $G$, corresponding to the ``pass'' announcement.

The key rate formula \eqref{eqnDevWin1} should be applied to the post-selected state:
\begin{align}
\label{eqnpostselect1}
\rhot_{AB} = \GC(\rho_{AB})/p_{\pass}
\end{align}
where $p_{\pass} = \Tr (\GC(\rho_{AB}))$ is the probability for passing the post-selection filter. We remark that since $\GC$ is given by a single Kraus operator, it maps pure states to pure states, and hence taking Eve's system to purify the post-selected state $\rhot_{AB}$ is equivalent to taking it to purify $\rho_{AB}$. Hence applying the key rate formula to $\rhot_{AB}$ does not give Eve access to any more than she already has, and hence does not introduce any looseness into our bound. Future extension of our work to more general maps $\GC$ will need to carefully account for how Eve's system is affected by $\GC$, so as not to lose key rate from this proof technique.

The only issue is that Alice's and Bob's experimental constraints $C$ in \eqref{eqnconstraintset} are still in terms of state $\rho_{AB}$. To solve for the key rate, one must transform these constraints into constraints on $\rhot_{AB}$. For the special case where $\GC$ has an inverse $\GC^{-1}$ that is also CP, one can simply insert the identity channel $\IC = \GC^{-1}\GC$ into the expression $\Tr(  \rho_{AB}\Gamma_i ) = \Tr( \GC^{-1}\GC (\rho_{AB})\Gamma_i )$. Using cyclic permutation under the trace, we transform \eqref{eqnconstraintset} into a set of constraints on $\rhot_{AB}$,  
\begin{align}
\label{eqnconstraintsetpostselect}
\tilde{C} = \{ \Tr(\rhot_{AB} \tilde{\Gamma}_i ) = \tilde{\gamma}_i\}.
\end{align}
where the $\tilde{\Gamma}_i = (\GC^{-1})\ad(\Gamma_i)$ are Hermitian operators, with $(\GC^{-1})\ad$ being the adjoint of $\GC^{-1}$, and $\tilde{\gamma}_i = \gamma_i / p_{\pass}$. Note that $p_{\pass}$ is determined experimentally and hence the $\tilde{\gamma}_i$ are known to Alice and Bob. More generally, we provide a method for obtaining $\tilde{C}$ for arbitrary $\GC$, as described in Supplementary Note~2.

We remark that public announcements can be treated with a simple extension of our post-selection framework. While our framework directly applies to announcements with only two outcomes corresponding to ``pass'' or ``fail'' (as discussed above), more general announcements can be treated by adding classical registers that store the announcement outcomes. Our treatment of MDI QKD is an example of this approach (see Fig.~\ref{fgrMDI} and Supplementary Note 1). Additional examples that could be treated in this way are protocols with two-way classical communication \cite{Gottesman2003} such as advantage distillation.

\begin{figure}
\begin{center}
\includegraphics[width=3.3in]{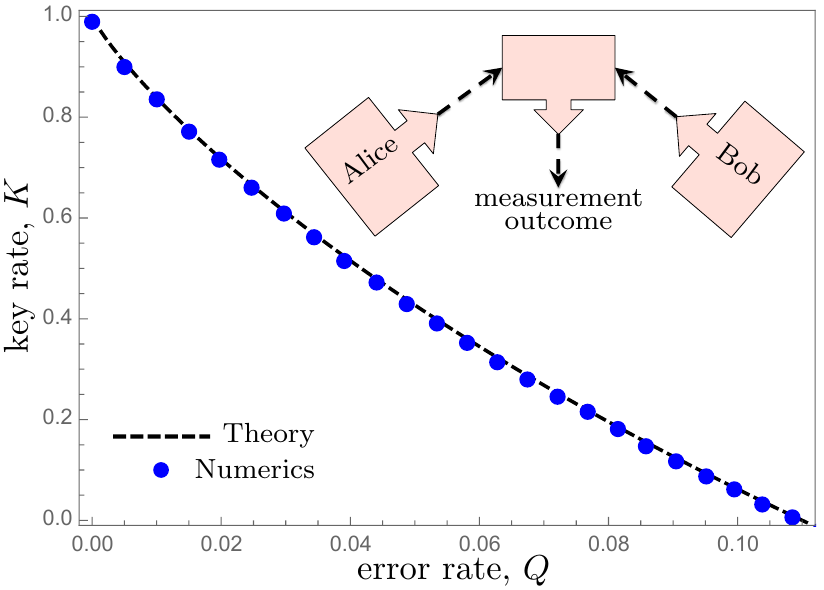}
\caption{Key rate for MDI QKD with the BB84 signal states. The inset shows the basic idea of MDI QKD: Alice and Bob each prepare a signal state and send it to an untrusted node, which performs an (untrusted) Bell-basis measurement and announces the outcome. Our numerics (circular dots) essentially reproduce the known theoretical dependence of the key rate on the error rate (dashed curve), which is the same expression as that given in \eqref{eqn2351}. See Supplementary Note~1 for elaboration.}
\label{fgrMDI}
\end{center}
\end{figure}

\bigskip

\noindent\textbf{Outline of examples.} We now illustrate our numerical approach for lower bounding the key rate by considering some well-known protocols. First, we consider the BB84 and six-state protocols (Fig.~\ref{fgrBB84sixstate}), MDI QKD with BB84 states (Fig.~\ref{fgrMDI}), and the generalized BB84 protocol involving two MUBs in any dimension (Fig.~\ref{fgrBB84highd}). In each case, the dependence of the key rate on error rate is known, and we show that our numerical approach exactly reproduces these theoretical dependences. After considering these structured protocols, we move on to using our numerical optimization for its intended purpose: studying unstructured protocols. The fact that our bound is tight for the structured protocols mentioned above gives reason to suspect that we will get strong bounds in the unstructured case. We investigate below a protocol involving $n$ MUBs, a protocol involving bases with arbitrary angle between them, and the B92 protocol.

\bigskip

\noindent\textbf{Example: BB84.} Consider an entanglement-based version of the BB84 protocol \cite{Bennett1984}, where Alice and Bob each receive a qubit and measure either in the $Z =\{\ket{0}, \ket{1}\}$ or $X = \{\ket{+}, \ket{-}\}$ basis, where $\ket{\pm}=(\ket{0}\pm \ket{1} )/\sqrt{2}$. For all protocols that we discuss, we assume perfect sifting efficiency, which can be accomplished asymptotically via asymmetric basis choice \cite{Lo2004}. Let us suppose that Alice and Bob each use their $Z$ basis in order to generate key. For simplicity, suppose they observe that their error rates in the $Z$ and $X$ bases are identical and equal to $Q$, then it is known (see, e.g., \cite{Scarani2009}) that the key rate is given by
\begin{align}
\label{eqn2351}
K = 1 - 2 h(Q)
\end{align}
where $h(p):= -p\log_2 p - (1-p)\log_2 (1-p)$ is the binary entropy.

To reproduce this result using our numerics, we write the optimization problem as follows: 
\begin{align}
\label{eqnBB84inputsA}\text{Key-map POVM:  }&Z_A = \{\dya{0},\dya{1}\}  \\
\label{eqnBB84inputsB}\text{Constraints:  }&\ave{\id} = 1 \\
\label{eqnBB84inputsC}&\ave{E_X} = Q\\
\label{eqnBB84inputsD}&\ave{E_Z} = Q
\end{align}
where the error operators are defined as
\begin{align}
\label{eqnEZdef}
E_Z &:= \dya{0}\ot\dya{1}+\dya{1}\ot\dya{0} \\
\label{eqnEXdef}
E_X &:= \dya{+}\ot\dya{-}+\dya{-}\ot\dya{+} .
\end{align}
Equations~\eqref{eqnBB84inputsA}-\eqref{eqnBB84inputsD} highlight the fact that, as far as the optimization in \eqref{eqnmainresult2} is concerned, a QKD protocol is defined by the POVM elements used for generating the key and the experimental constraints used for ``parameter estimation'' (and also the post-selection map $\GC$, but this is trivial for the ideal BB84 protocol.). Once these things are specified, the protocol is defined and the key rate is determined. Numerically solving the problem defined in \eqref{eqnBB84inputsA}-\eqref{eqnBB84inputsD} for several values of $Q$ leads to the red dots in Fig.~\ref{fgrBB84sixstate}, which agree perfectly with the theory curve.

\bigskip

\noindent\textbf{Example: Six state.} Now consider an entanglement-based version of the six-state protocol, where Alice and Bob each measure one of three MUBs ($X$, $Y$, or $Z$) on their qubit. Suppose that Alice and Bob observe that their error rates in all three bases are identical, $\ave{E_X} = \ave{E_Y} =\ave{E_Z} = Q$, where
\begin{align}
E_Y &:= \dya{y_+}\ot\dya{y_+}+\dya{y_-}\ot\dya{y_-}\,,
\end{align}
with $\ket{y_{\pm}} = (\ket{0} \pm i \ket{1})/\sqrt{2}$. (Our definition of $E_Y$ reflects the fact that the standard Bell state is correlated in $Z$ and $X$ but anti-correlated in $Y$.) To reproduce the known key rate \cite{Bruss2002,Renner2005a}, we write the optimization problem as: 
\begin{align}
\label{eqnsixstateinputsA}\text{Key-map POVM:  }&Z_A = \{\dya{0},\dya{1}\}  \\
\label{eqnsixstateinputsB}\text{Constraints:  }&\ave{\id} = 1 \\
\label{eqnsixstateinputsC}&\ave{E_{XY}} = Q\\
\label{eqnsixstateinputsD}&\ave{E_Z} = Q \,,
\end{align}
where $E_{XY} := (E_X+E_Y)/2$ quantifies the average error for $X$ and $Y$. Note that the constraint $\ave{E_{XY}} = Q$ is obtained by coarse-graining the individual error rates. In theory, one can get a stronger bound on the key rate by splitting up this constraint into $\ave{E_{X}} = Q$ and $\ave{E_{Y}} = Q$. However, our numerics show that this does not improve the key rate, and the constraints in \eqref{eqnsixstateinputsB}-\eqref{eqnsixstateinputsD} are enough to reproduce the theory curve. Indeed, numerically solving the problem in \eqref{eqnsixstateinputsA}-\eqref{eqnsixstateinputsD} leads to the blue dots in Fig.~\ref{fgrBB84sixstate}, which agree with the theory curve.

\bigskip

\noindent\textbf{Example: Two MUBs in higher dimensions.} A distinct advantage of our approach of solving \eqref{eqnmainresult2} instead of the primal problem \eqref{eqnprimalproblem5} is that we can easily perform the optimization in higher dimensions, where the number of parameters in \eqref{eqnprimalproblem5} would be quite large. To illustrate this, we consider a generalization of BB84 to arbitrary dimension, where Alice and Bob measure generalized versions of the $X$ and $Z$ bases. This protocol has been implemented, e.g., in Ref.~\cite{Mafu2013} using orbital angular momentum. Taking $Z$ as the standard basis $\{\ket{j}\}$, Alice's $X$ basis can be taken as the Fourier transform $\{F\ket{j}\}$, where
\begin{align}
\label{eqnfouriermatrix231}
F = \sum_{j,k} \frac{\om^{-j k}}{\sqrt{d}}\dyad{j}{k}
\end{align}
is the Fourier matrix, with $\om = e^{2\pi i/d}$, and for simplicity we choose Alice's and Bob's dimension to be equal: $d_A = d_B = d$. Bob's $X$ basis is set to $\{F^*\ket{j}\}$, where $F^*$ denotes the conjugate of $F$ in the standard basis.

\begin{figure}[t]
\begin{center}
\includegraphics[width=3.3in]{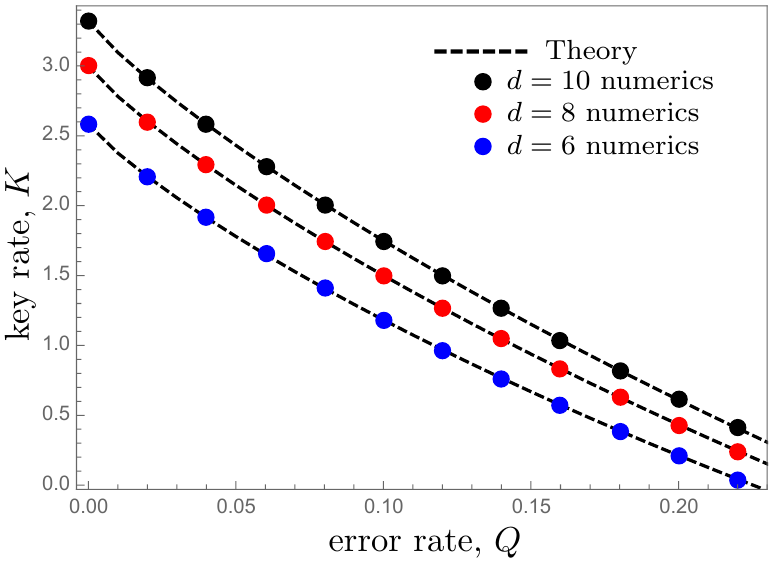}
\caption{Higher dimensional analog of BB84, using two MUBs. This plot shows the theoretical key rate as solid curves, and the result of our numerical optimization as circular dots, for $d_A=d_B = d$, with $d=6$ (blue), $d=8$ (red), and $d=10$ (black).  Again, the dots should be viewed as reliable lower bounds, but in this case they are perfectly tight.}
\label{fgrBB84highd}
\end{center}
\end{figure}

Suppose that Alice and Bob observe that their error rates in $Z$ and $X$ are identical. The theoretical key rates \cite{Sheridan2010,Ferenczi2012} for the cases $d=6,8,10$ are shown as dashed curves in Fig.~\ref{fgrBB84highd}, while our numerics are shown as circular dots. Clearly there is perfect agreement with the theory.

For our numerics we employ the same constraints as used for BB84 in \eqref{eqnBB84inputsA}-\eqref{eqnBB84inputsD}, but generalized to higher $d$. We again emphasize that the calculation of $\Th$ here is very efficient and can easily handle higher dimension. This is because the number of parameters one is optimizing over is independent of dimension - equal to the number of constraints, which in this case is 3. This is in sharp contrast to the primal problem in \eqref{eqnprimalproblem5}, where the number of parameters is $d^4$, which would be 10000 for $d=10$.

\bigskip

\bigskip

\noindent\textbf{Example: $n$ MUBs.} A simple generalization of the above protocols is to consider a set of $n$ MUBs in dimension $d$. For example, in prime power dimensions there exist explicit constructions for sets of $n$ MUBs with $2 \leq n  \leq d+1$ \cite{Bandyopadhyay2001}. Consider a protocol where we fix the set of $n$ MUBs, and in each round, Alice and Bob each measure their $d$ dimensional system in one basis chosen from this set. For general $n$ the symmetry group is not known for this protocol \cite{Ferenczi2012}, so one can consider it an unstructured protocol. Indeed, only for the special cases $n=2$ and $n=d+1$ do we have analytical formulas for the key rate \cite{Ferenczi2012}. Nevertheless it is straightforward to apply our numerics to this protocol for any $n$. Our results are shown in Fig.~\ref{fgrBB84nmubs} for $d=5$. To obtain these curves we only need three constraints, which are analogous to \eqref{eqnsixstateinputsB}-\eqref{eqnsixstateinputsD}, but generalized such that $\ave{E_{XY}}$ is replaced by the average error rate in all $n-1$ bases, excluding the basis used for key generation (the $Z$ basis).

\begin{figure}[t]
\begin{center}
\includegraphics[width=3.25in]{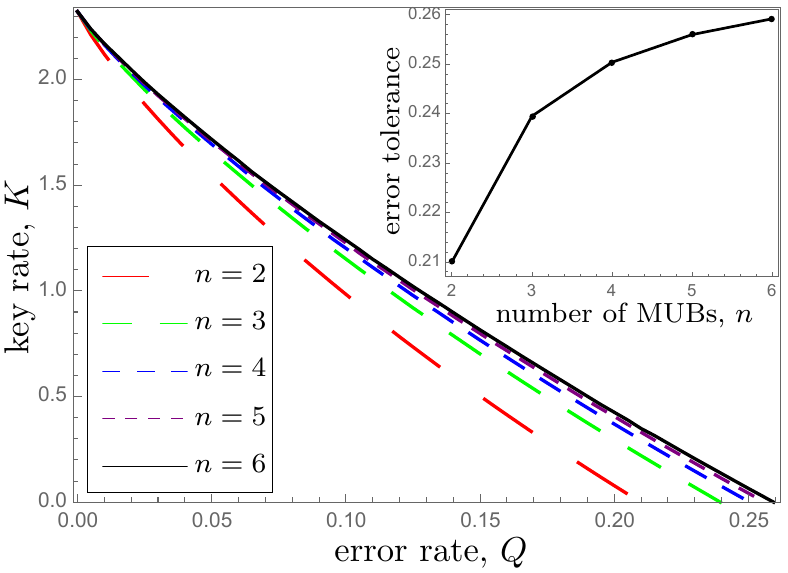}
\caption{Protocol where Alice and Bob each use $n$ MUBs. The key rate is plotted for various $n \in \{2, 3, 4, 5, 6\}$ and for $d_A=d_B = 5$. This is an unstructured protocol, since for intermediate values of $n$ the symmetry group and hence the key rate is unknown. However, our numerics provides the dependence of key rate on error rate for any $n$, as shown. The inset shows the error tolerance - the smallest error rate that makes the key rate vanish - as a function of $n$. Note that the largest jump in the error tolerance occurs from $n=2$ to $n=3$.}
\label{fgrBB84nmubs}
\end{center}
\end{figure}

Interestingly, Fig.~\ref{fgrBB84nmubs} shows that just adding one basis, going from $n=2$ to $n=3$, gives a large jump in the key rate, whereas there are diminishing returns as one adds more bases. This can be seen in the inset of Fig.~\ref{fgrBB84nmubs}, which plots the error tolerance (i.e., the value of $Q$ for which the key rate goes to zero) as a function of $n$. We have seen similar behavior for other $d$ besides $d=5$. After completion of this work, an analytical formula for $n=3$ was discovered \cite{Bradler2015}, and we have verified that it agrees perfectly with our numerics.

In Supplementary Note~3, we analytically prove the following.
\vspace{3pt}
\begin{proposition}
\label{prop1}
Our numerical results are perfectly tight for the protocols discussed in Figs.~\ref{fgrBB84sixstate}, \ref{fgrBB84highd}, and \ref{fgrBB84nmubs}. That is, for these protocols, our optimization exactly reproduces the primal optimization \eqref{eqnprimalproblem5}.
\end{proposition}
Note that this observation implies that key rate for protocols involving $n$ MUBs (as in Fig.~\ref{fgrBB84nmubs}) is now known; namely it is given by our numerical optimization.

\bigskip

\noindent\textbf{Example: Arbitrary angle between bases.} While MUBs are a special case, our approach can handle arbitrary angles between the different measurements or signal states. For example, we consider a simple qubit protocol \cite{Matsumoto2009a} where Alice and Bob each measure either the $Z$ or $W$ basis, where $W$ is rotated by an angle $\theta$ away from the ideal $X$ basis. This protocol provides the opportunity to compare our numerical approach to an analytical approach based on the entropic uncertainty principle, introduced in Refs.~\cite{Berta2010,Tomamichel2012a}. This is the state-of-the-art method for lower bounding the key rate. So for comparison, Fig.~\ref{fgrBB84vstheta} plots the bound obtained from the entropic uncertainty principle for bases $Z$ and $W$.

\begin{figure}[t]
\begin{center}
\includegraphics[width=3.45in]{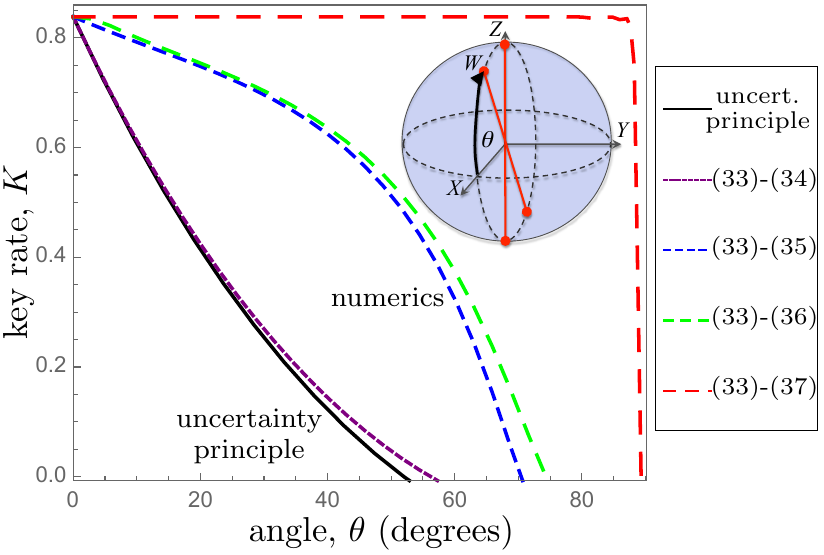}
\caption{Protocol where Alice and Bob each measure $Z$ or $W$. Here $Z$ is the standard basis and the $W$ basis is rotated by an angle $\theta$ away from the $X$ basis. The key rate versus $\theta$ is shown with the error rate set to $Q = 0.01$. Our numerics give a hierarchy of four lower bounds on the key rate, corresponding to adding in additional constraints from \eqref{eqnbb84vsthetaA}-\eqref{eqnbb84vsthetaE}. All of our bounds are tighter than the bound obtained from the entropic uncertainty principle. The plot indicates that the uncertainty principle gives a dramatically pessimistic key rate, much lower than the true key rate of the protocol.}
\label{fgrBB84vstheta}
\end{center}
\end{figure}

We apply our numerical approach with the constraints:
\begin{align}
\label{eqnbb84vsthetaA}\text{Constraints:  }&\ave{\id} = 1 \\
\label{eqnbb84vsthetaB}&\ave{\sigma_W\ot \sigma_W} = 1-2 Q \\
\label{eqnbb84vsthetaC}&\ave{\sigma_Z\ot \sigma_Z} = 1-2 Q\\
\label{eqnbb84vsthetaD}&\ave{\sigma_Z\ot \sigma_W} =(\sin\theta )(1-2 Q)\\
\label{eqnbb84vsthetaE}&\ave{\sigma_W\ot \sigma_Z} = (\sin\theta )(1-2 Q),
\end{align}
where $\sigma_Z$ and $\sigma_W$ are the Pauli operators associated with the $Z$ and $W$ bases. Fig.~\ref{fgrBB84vstheta} plots a hierarchy of lower bounds obtained from gradually adding in more of the constraints in \eqref{eqnbb84vsthetaA}-\eqref{eqnbb84vsthetaE}. As the plot shows, we already beat the entropic uncertainty principle with only the first two constraints. Furthermore, adding in all these constraints gives a dramatically higher bound, showing the uncertainty principle gives highly pessimistic key rates for this protocol. From an experimental perspective, Fig.~\ref{fgrBB84vstheta} is reassuring, in that small variations in $\theta$ away from the ideal BB84 protocol ($\theta =0$) have essentially no effect on the key rate. Fig.~\ref{fgrBB84vstheta} also highlights the fact that our approach allows us to systematically study the effect on the key rate of Alice and Bob using more or less of their data. In this example, we see that it is useful to keep data that one will usually discard in the sifting step of the protocol.

\bigskip

\noindent\textbf{Example: B92.} Next we consider the B92 protocol \cite{Bennett1992}, which is a simple, practical, unstructured protocol. It nicely illustrates our framework because it is inherently a prepare-and-measure protocol and it involves post-selection. In the protocol, Alice sends one of two non-orthogonal states $\{\ket{\phi_0} , \ket{\phi_1}\}$ to Bob. Since the Bloch-sphere angle $\theta$ between the two states is arbitrary, with $\ip{\phi_0}{\phi_1} = \cos (\theta/2)$, the protocol is unstructured. Bob randomly (with equal probability) measures either in basis $B_0 = \{\ket{\phi_0},\ket{\overline{\phi_0}}\}$ or basis $B_1 = \{\ket{\phi_1},\ket{\overline{\phi_1}}\}$, where $\ip{\phi_0}{\overline{\phi_0}}=\ip{\phi_1}{\overline{\phi_1}}=0$. If Bob gets outcome $\ket{\overline{\phi_0}}$ or $\ket{\overline{\phi_1}}$, then he publicly announces ``pass'', and he assigns a bit value of 1 or 0, respectively, to his key. Otherwise, Bob announces ``fail'' and they discard the round.

\begin{figure}
\begin{center}
\includegraphics[width=3.34in]{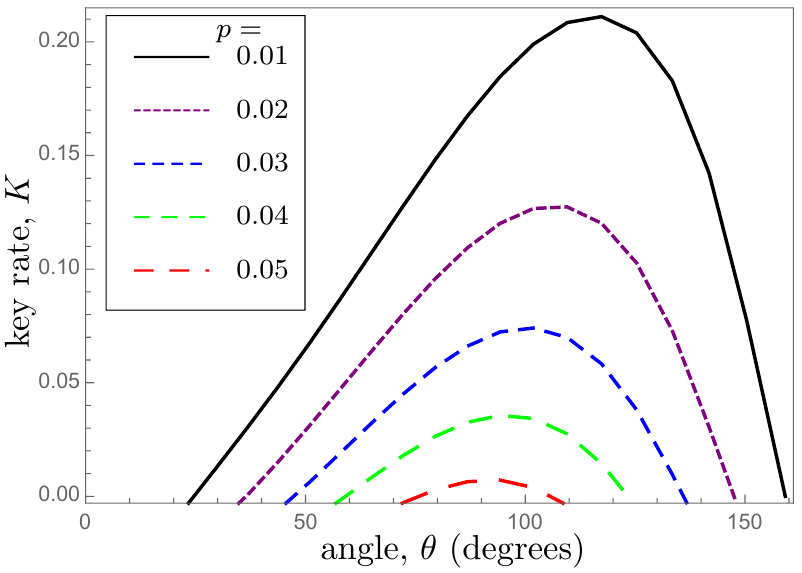}
\caption{The B92 protocol. The key rate (in bits per photon sent by Ailce) is plotted versus the Bloch-sphere angle between the two signal states. Curves are shown for various values of the depolarizing probability $p$.}
\label{fgrB92}
\end{center}
\end{figure}

A detailed description of the constraints we employed for B92 can be found in Supplementary Note~4. Our numerical results are shown in Fig.~\ref{fgrB92}. Fig.~\ref{fgrB92} shows that the optimal angle for maximizing key rate depends on the depolarizing noise $p$, although small deviations $\pm 5^\circ$ from the optimal angle do not affect the key rate much.

Our results give higher key rates for B92 than Refs.~\cite{Tamaki2004} and \cite{Renner2005a}, which respectively predicted positive key rates for $p\leq 0.034$ and $p\leq 0.048$, while we predict it for $p\leq 0.053$. On the other hand, Ref.~\cite{Matsumoto2013} directly solved the primal problem \eqref{eqnprimalproblem5} for B92 by brute-force numerics, and achieves positive key rate for $p\leq 0.065$. We have verified that the gap between our results and those of Ref.~\cite{Matsumoto2013} is due to the looseness of our usage of the Golden-Thompson inequality (see Eq.~\eqref{eqn235632983} in the Methods section). However, Ref.~\cite{Matsumoto2013} only showed a plot for $p\geq 0.046$, noting that the numerical optimization for the primal problem did not converge for smaller $p$ values. This highlights a benefit of going to the dual problem, in that we have no trouble with obtaining the full dependence on $p$.

\section*{\uppercase{Discussion}}

In conclusion, we address one of the main outstanding problems in QKD theory: how to calculate key rates for arbitrary protocols. Our main result is a numerical method for lower-bounding key rates that is both efficient and reliable. It is reliable in the sense that, by reformulating the problem as a maximization, every solution that one's computer outputs is an achievable key rate. It is efficient in the sense that we have reduced the number of parameters in the optimization problem from $d_A^2d_B^2$ down to the number of experimental constraints, which in some cases is independent of dimension.

The motivation for our work is two-fold. First, experimental imperfections tend to break symmetries, so theoretical techniques that exploit symmetries do not apply. Hence there is no general method currently available for calculating the effect of imperfections on the key rate. Second, it is interesting to ask whether protocols that are intentionally designed to lack symmetry might outperform the well-known symmetric protocols. Such a question cannot be posed without a method for calculating key rates for unstructured protocols. Just to give an example where the key rate is currently unknown, we plan to apply our approach to protocols where a small, discrete set of coherent states are the signal states and information is encoded in the phase~\cite{Lo2007}.

We envision that our method could be a standard tool for QKD researchers. In future work we hope to extend our approach to the finite-key scenario. Indeed, the optimization problem we solve is closely related to one appearing in finite-key analysis~\cite{Scarani2008}.

\section*{\uppercase{Methods}}\label{sctderivation}

\noindent\textbf{Outline.} Here we prove our main result, Theorem~\ref{thm1}. Our proof relies on several technical tools. First is the notion of the duality of optimization, i.e., transforming the primal problem to its dual problem. Second, we employ several entropic identities to simplify the dual problem. Third, we use a recent, important result from Ref.~\cite{Zorzi2014} that solves a relative entropy optimization problem.

For readability, we prove Theorem~\ref{thm1} here for the special case where the key-map POVM $Z_A = \{Z_A^j\}$ is a projective measurement, i.e., where the $Z_A^j$ are projectors (of arbitrary rank). We postpone the proof for arbitrary POVMs to Supplementary Note~5. 

\bigskip

\noindent\textbf{The primal problem.} First we rewrite \eqref{eqnprimalproblem5} as:
\begin{align}
\label{eqnprimalrestated4644}
K &= \left[\min_{\rho_{AB}\in \CC} H(Z_A | E) \right] - H(Z_A | Z_B) \,,
\end{align}
noting that the second term in \eqref{eqnprimalrestated4644}, $H(Z_A | Z_B)$, will be determined experimentally and hence can be pulled out of the optimization. We remark that, simply for illustration purposes we used Fano's inequality to upper-bound $H(Z_A | Z_B)$ in our figures; however, in practice $H(Z_A | Z_B)$ would be directly calculated from the data.

Since we only need to optimize the first term, we redefine the primal problem as
\begin{align}
\label{eqnprimalrestated4649}
\alpha &:= \min_{\rho_{AB}\in \CC} H(Z_A | E) \,,
\end{align}
and note that we can take $E$ to be a purifying system of $\rho_{AB}$, since that gives Eve the most information. Next we use a result for tripartite pure states $\rho_{ABE} = \dya{\psi}_{ABE}$ from Refs.~\cite{Coles2011,Coles2012} that relates the conditional entropy to the relative entropy:
\begin{align}
\label{eqnconrelentropyproj}
H(Z_A | E)= D\bigg(\rho_{A B} \bigg|\bigg| \sum_j Z_{A}^j \rho_{A B} Z_{A}^j\bigg)
\end{align}
where the relative entropy is defined by
\begin{align}
\label{eqnrelentropydef}
D(\sg || \tau):= \Tr(\sg\log_2 \sg)-\Tr(\sg\log_2 \tau).
\end{align}
We remark that the joint convexity of the relative entropy implies that the right-hand side of \eqref{eqnconrelentropyproj} is a convex function of $\rho_{AB}$. (See \cite{Watanabe2008} for an alternative proof of convexity.) Because of this, and the fact that the constraints in \eqref{eqnconstraintset} are linear functions of $\rho_{AB}$, \eqref{eqnprimalrestated4649} is a convex optimization problem \cite{Boyd2010}.  

It is interesting to point out the connection to coherence \cite{Baumgratz2013a}. For some set of orthogonal projectors $\Pi=\{\Pi^j \}$ that decompose the identity, $\sum_j \Pi^j = \id$, the coherence (sometimes called relative entropy of coherence) of state $\rho$ is defined as \cite{Baumgratz2013a}:
\begin{align}
\label{eqncoherencedef}
\Phi(\rho, \Pi) = D\bigg(\rho \bigg|\bigg|  \sum_j \Pi^j \rho \Pi^j \bigg)\,.
\end{align}
Rewriting the primal problem in terms of coherence gives
\begin{align}
\label{eqnfirsttermprimal2}
\alpha = \min_{\rho_{AB}\in \CC} \Phi(\rho_{AB}, Z_A)\,.
\end{align}
Hence we make the connection that calculating the secret key rate is related to optimizing the coherence.

This observation is important since the coherence is a continuous function of $\rho$ (see Supplementary Note~6). This allows us to argue in Supplementary Note~6 that our optimization problem satisfies the strong duality criterion \cite{Boyd2010}, which means that the solution of the dual problem is precisely equal to that of primal problem.

\bigskip

\noindent\textbf{The dual problem.} Now we transform to the dual problem. Due to a pesky factor of $\ln (2)$, it is useful to rescale the primal problem as follows:
\begin{align}
\label{eqnprimalrestatedhat}
\widehat{\al} := \al \ln(2) = \min_{\rho_{AB}\in \CC} \widehat{\Phi}(\rho_{AB}, Z_A)
\end{align}
where, henceforth, we generally use the notation $\widehat{M}:=M \ln(2)$, for any quantity $M$. The dual problem \cite{Boyd2010} of \eqref{eqnprimalrestatedhat} is given by the following unconstrained optimization:
\begin{align}
\label{eqnfirsttermdual}
\widehat{\beta} =  \max_{\lmv} \min_{\rho_{AB}\in \PC }\LC(\rho_{AB}, \lmv)
\end{align}
where $\PC$ is the set of positive semidefinite operators:
\begin{align}
\label{eqnpositiveset23091}
\PC = \{\rho_{AB} \in \HC_{d_A d_B} : \rho_{AB}\geq 0 \}.
\end{align}
Here the Lagrangian is given by
\begin{align}
\label{eqnlagrangian}
\LC(\rho_{AB}, \lmv):= \widehat{\Phi}(\rho_{AB}, Z_A)  +\sum_i \lambda_i [\Tr(\rho_{AB}\Gamma_i ) - \gamma_i] \,,
\end{align}
where the $\lmv = \{\lm_i\}$ are Lagrange multipliers. Strong duality implies that 
\begin{align}
\label{eqnstrongdual}
\widehat{\beta} =\widehat{\al} \,.
\end{align}

In what follows, we go through several steps to simplify the dual problem. It helps to first state the following lemma.
\begin{lemma}\cite{Modi2010,Coles2012}
\label{lemcoherenceopt}
For any $\rho$ and $\Pi=\{ \Pi^j \}$, the coherence can be rewritten as
\begin{align}
\label{eqnrelentropyidentity}
\Phi(\rho, \Pi) =\min_{\omega \in \DC}D\bigg(\rho \bigg|\bigg|  \sum_j \Pi^j \omega \Pi^j \bigg) 
\end{align}
where $\DC$ is the set of density operators.
\end{lemma}
Hence we have
\begin{align}
\label{eqn23532625}
\widehat{\Phi}(\rho_{AB}, Z_A)&= \min_{\sg_{ A B}\in\DC} \widehat{D}\big(\rho_{A B} || \ZC_A (\sg_{A B} )\big) \,,
\end{align}
where we define the quantum channel $\ZC_A$ whose action on an operator $O$ is given by
\begin{align}
\label{eqnzadef}
\ZC_A(O):=\sum_j Z_{A}^j O Z_{A}^j\,.
\end{align}
Next, we interchange the two minimizations in \eqref{eqnfirsttermdual} 
\begin{align}
\label{eqn2356797}
\min_{\rho_{AB}\in\PC}\min_{\sg_{ A B}\in\DC} f(\rho_{AB}, \sg_{AB}, \vec{\lm}) \notag\\
= \min_{\sg_{AB}\in\DC}\min_{\rho_{ A B}\in\PC} f(\rho_{AB}, \sg_{AB}, \vec{\lm})
\end{align}
where
\begin{align}
\label{eqnfdef1}
f(\rho_{AB}, \sg_{AB}, \vec{\lm}) := \widehat{D}&\big(\rho_{A B} || \ZC_A ( \sg_{A B} )\big)\notag\\
&+\sum_i \lambda_i (\ave{\Gamma_i}-\gamma_i). 
\end{align}

Ref.~\cite{Zorzi2014} solved a relative entropy optimization problem, a special case of which is our problem:
\begin{align}
\label{eqn2356798}
\min_{\rho_{ A B}\in\PC} f(\rho_{AB}, \sg_{AB}, \vec{\lm})\,.
\end{align}
From \cite{Zorzi2014}, the unique solution of \eqref{eqn2356798} is
\begin{align}
\label{eqn2356799}
\rho_{AB}^* = \exp \Big( -\id - \lmv\cdot \Gmv+ \ln \big(\ZC_A ( \sg_{A B} ) \big)\Big)\,.
\end{align}
Inserting \eqref{eqn2356799} into \eqref{eqnfdef1} gives the optimal value:
\begin{align}
\label{eqn23373}
f(\rho_{AB}^*, \sg_{AB}, \lmv) = - \Tr (\rho_{AB}^* )- \sum_i \lambda_i \gamma_i\,.
\end{align}
In summary the dual problem becomes
\begin{align}
\label{eqndualrewritten}
\widehat{\beta} = & \max_{\lmv} \eta(\lmv)\,,
\end{align}
with
\begin{align}
\label{eqnetadef1}
\eta(\lmv):= - \max_{\sg_{AB}\in\DC} \left[ \Tr (\rho_{AB}^* ) + \lmv\cdot \vec{\gamma}\right]\,.
\end{align}

\bigskip

\noindent\textbf{A lower bound.} We can obtain a simple lower bound on $\eta(\lmv)$ as follows. The Golden-Thompson inequality states that
\begin{align}
\label{eqn235632982}
\Tr(\exp(A+B))\leq \Tr(\exp(A)\exp(B)).
\end{align}
Applying this inequality gives:
\begin{align}
\label{eqn235632983}
\Tr (\rho_{AB}^* ) &\leq \Tr \Big(R(\lmv) \exp \big( \ln \ZC_A( \sg_{A B}) \big)\Big)\\
 &= \Tr \Big(R(\lmv)  \ZC_A (\sg_{A B} ) \Big) \\
 &= \Tr \Big( \ZC_A \big(  R(\lmv) \big) \sg_{A B} \Big)\,,
\end{align}
where $R(\lmv) = \exp \big(-\id - \lmv\cdot\Gmv \big)$ was defined in \eqref{eqnmainresult3}. Next, note that
\begin{align}
\label{eqn235632984}
 \max_{\sg_{AB}\in\DC} \Tr \Big( \ZC_A \big(   R(\lmv)   \big) \sg_{A B} \Big)= \Big\| \ZC_A \big(   R(\lmv)   \big) \Big\| \,.
\end{align}
Hence, we arrive at our final result
\begin{align}
\label{eqn235632985}
\widehat{\beta} & \geq  \max_{\lmv}  \left[ -  \Big\| \ZC_A \big(    R(\lmv)   \big) \Big\|  - \lmv\cdot \vec{\gamma} \right] \,,
\end{align}
where the right-hand side is denoted as $\Th$ in Theorem~\ref{thm1}.

\section*{\uppercase{Acknowledgements}}\label{sctackknowledge}

We thank Jie Lin, Adam Winick, and Bailey Gu for technical help with the numerics and for obtaining the data in Fig.~\ref{fgrMDI}. We thank Yanbao Zhang and Saikat Guha for helpful discussions. We acknowledge support from Industry Canada, Sandia National Laboratories, Office of Naval Research (ONR), NSERC Discovery Grant, and Ontario Research Fund (ORF).

\section*{\uppercase{Author Contributions}}\label{sctcontributions}

P.C.\ obtained the conceptional main results. E.M.\ contributed the approach to incorporate post-selection and announcements. N.L.\ conceived and supervised the project. P.C.\ wrote the manuscript with input from E.M.\ and N.L.

\bibliographystyle{naturemag}
\bibliography{uqkd}

\newpage

\onecolumngrid

\section*{\uppercase{Supplementary Note 1: MDI QKD}}\label{appmdi}

In the Results section, we outlined our framework for handling MDI QKD protocols. Here we elaborate on this framework, and we also give more details on the example calculation shown in Fig.~2.

\subsection*{Framework for MDI QKD (continued)}

Our framework considers the tripartite state $\rho_{ABM}$, where $A$ and $B$ are respectively the systems held by Alice and Bob in the source-replacement scheme, and $M$ is the classical register that stores the outcome of the measurement performed by the untrusted node. Let us elaborate on the origin of $\rho_{ABM}$. Recall that, in the source-replacement scheme, Alice prepares a bipartite entangled state of the form
\begin{align}
\label{eqnmdiapp1}
\ket{\psi_{AA'}}= \sum_j \sqrt{p_j}\ket{j}\ket{\phi_j},
\end{align}
and in the MDI scenario, Bob prepares a similar state
\begin{align}
\label{eqnmdiapp2}
\ket{\psi_{BB'}}= \sum_j \sqrt{p_j}\ket{j}\ket{\phi_j}\,.
\end{align}
Hence, the initial state (prior to the action of Eve) is
\begin{align}
\label{eqnmdiapp3}
\rho^{(0)}_{AA'BB'}:= \dya{\psi_{AA'}}\ot \dya{\psi_{BB'}} \,.
\end{align}
For notational convenience, it is helpful to permute the order of the subsystems, as follows
\begin{align}
\label{eqnmdiapp4}
\rhot^{(0)}_{ABA'B'} := \FC \Big(\rho^{(0)}_{AA'BB'} \Big)  \,,
\end{align}
where $\FC$ is the quantum channel that switches the ordering of subsystems $A'$ and $B$. 

Now note that Eve only has access to $A'B'$ and not $AB$. Likewise the untrusted node performs a measurement only on $A'B'$, while $A$ and $B$ remain respectively in Alice's and Bob's laboratories. We combine the action of Eve together with the action of the untrusted measurement, and model it as a single quantum channel $\EC$ that maps $A'B' \to M$, where $M$ is a classical register. That is, we obtain the state
\begin{align}
\label{eqnmdiapp5}
\rho_{ABM} = \Big(\IC_{AB} \ot \EC\Big)\Big(\rhot^{(0)}_{ABA'B'}\Big)  \,,
\end{align}
where $\IC_{AB}$ is the identity channel on $AB$. We apply our numerical approach to the state $\rho_{ABM}$ in Supplementary Eq.~\eqref{eqnmdiapp5}. The beauty of the MDI protocol is that we do not need to consider the process of how we arrived at the state $\rho_{ABM}$, i.e., we do not need to discuss the details of the channel~$\EC$. We only need to specify the experimental constraints on $\rho_{ABM}$, which we stated in the Results section (although we repeat them here for convenience), 
\begin{align}
\label{eqnmdiapp6}
\Tr \left[\rho_{ABM } (\dya{j}\ot \dya{k} \ot \dya{m}  ) \right] = \gamma_{jkm}\,.
\end{align}
In addition, we also enforce constraints that fix the form of the marginal $\rho_{AB}$, which has the form
\begin{align}
\label{eqnmdiapp7}
\rho_{AB} = \rho_A \ot \rho_B = \left( \sum_{j,k} \sqrt{p_j p_k} \ip{\phi_k}{\phi_j}\dyad{j}{k} \right) \ot \left(\sum_{j,k} \sqrt{p_j p_k} \ip{\phi_k}{\phi_j}\dyad{j}{k}\right) 
\,.
\end{align}

One could also add constraints that enforce that $M$ is a classical system. However, we choose not to do this for the following reason. The worst-case scenario, i.e., the scenario that gives Eve the most information, corresponds to $M$ being classical, and hence the key rate is not improved by enforcing the classicality of $M$. We state this in the following lemma.

\begin{lemma}\label{lemmamdi}
Let $\{\ket{m}\}$ be the standard basis for system $M$, and let $\MC$ be the quantum channel that diagonalizes (i.e., decoheres) system $M$ in this basis. That is, $\MC(O) = \sum_m \dya{m} O \dya{m}$ for any operator~$O$. Consider a set of constraints $C$ on $\rho_{ABM}$ and let $\CC$ denote the set of density operators $\rho_{ABM}$ that satisfy $C$. Suppose that the constraints $C$ do not preclude $M$ from being decohered in the standard basis, i.e.,  if $\rho_{ABM} \in \CC$, then $(\IC_{AB}\ot \MC)(\rho_{ABM}) \in \CC$. Define the set
\begin{align}
\label{eqnmdiapp8}
\CC_{\MC}:= \{ \rho_{ABM} \in \CC : (\IC_{AB}\ot \MC)(\rho_{ABM}) = \rho_{ABM}  \}\,.
\end{align}
In other words, $\CC_{\MC} \subseteq \CC$ is the set of states in $\CC$ that are diagonal in the standard basis on $M$. Then, Eve's ignorance about Alice's key is the same regardless of whether we impose that $M$ is decohered in the standard basis, i.e., 
\begin{align}
\label{eqnmdiapp9}
\min_{\rho_{ABM} \in \CC}H(Z_A|E) = \min_{\rho_{ABM} \in \CC_{\MC}}H(Z_A|E)\,.
\end{align}
\end{lemma}
\begin{proof}
For notational simplicity we drop the subscript $A$ from $Z_A$ in what follows. Since $\CC_{\MC} \subseteq \CC$, then we obviously have 
\begin{align}
\label{eqnmdiapp10}
\min_{\rho_{ABM} \in \CC}H(Z|E) \leq \min_{\rho_{ABM} \in \CC_{\MC}}H(Z|E)\,,
\end{align}
so we just need to show the inequality in the opposite direction. In particular we will show that for each state in $\CC$ there is a corresponding state in $\CC_{\MC}$ where Eve's ignorance is lower. Let $\rho_{ABM} \in \CC$, then 
\begin{align}
\label{eqnmdiapp11}
\rhot_{ABM} := (\IC_{AB}\ot \MC)(\rho_{ABM})\in \CC_{\MC}\,.
\end{align}
Let $E$ and $\tilde{E}$ be purifying systems for $\rho_{ABM}$ and $\rhot_{ABM}$, respectively. Then the states
\begin{align}
\label{eqnmdiapp12}
\sigma_{ZZ'ABME} &:= (V\ot \id_{BME})\rho_{ABME}(V\ad\ot \id_{BME})\,,\quad\text{and}\\
 \tilde{\sigma}_{ZZ'ABM\tilde{E}} &:=(V\ot \id_{BM\tilde{E}})\rhot_{ABM\tilde{E}}(V\ad\ot \id_{BM\tilde{E}})
\end{align}
are pure states. Here, $V$ is an isometry that maps $A\to ZZ'A $, defined by
\begin{align}
\label{eqnmdiapp13}
V:=\sum_j \ket{j}_{Z}\ot \ket{j}_{Z'}\ot \sqrt{Z_A^j}\,,
\end{align}
where the set $\{Z_A^j\}$ forms a POVM (Alice's key-map POVM). 

Let us take a moment to clarify the meaning of the conditional entropy $H(Z|E)$. Note that, by convention, when we casually refer to $H(Z |E)$ for the state $\rho_{ABME}$, we precisely mean the conditional von Neumann entropy of the state $\sigma_{ZZ'ABME}$, which we denote $H(Z |E)_{\sigma}$. Typically one refers to $\sigma_{ZZ'ABME}$ as the post-measurement state associated with a given (pre-measurement) state $\rho_{ABME}$. Likewise $H(Z |\tilde{E})$ for the state $\rhot_{ABM\tilde{E}}$ actually refers to the conditional von Neumann entropy of $\tilde{\sigma}_{ZZ'ABM\tilde{E}} $ denoted by $H(Z | \tilde{E})_{\tilde{\sigma}}$. 

The duality \cite{Konig2009} of the von Neumann entropy says that $H(A|B)_{\tau} = - H(A|C)_{\tau}$ for any tripartite pure state $\tau_{ABC}$. Applying this duality relation to the pure state $\sigma_{ZZ'ABME}$ gives
\begin{align}
\label{eqnmdiapp14}
H(Z |E)_{\sigma} &= - H(Z |Z'ABM)_{\sigma} \\
&\geq - H(Z|Z'ABM)_{\tilde{\sigma}}\\
&=H(Z | \tilde{E})_{\tilde{\sigma}} \,,
\end{align}
where the inequality is due to the data-processing inequality, i.e., acting with channel $\MC$ on $M$ can never reduce the entropy. Hence we have shown that Eve's ignorance for the state $\rhot_{ABM\tilde{E}}$ is not larger than her ignorance for the state $\rho_{ABME}$, which is the desired result. 
\end{proof}

\subsection*{Example: MDI QKD with BB84 states}

Here we elaborate on how we obtain the data in Fig.~2. To obtain this data, we consider the most common MDI protocol, where Alice and Bob each prepare and send the BB84 signal states $\{\ket{0},\ket{1},\ket{+},\ket{-}\}$ with probabilities $p_z /2$ and $(1-p_z)/2$ respectively for the $Z$- and $X$-basis states. For simplicity we consider a protocol that does not do sifting and distills key out of both the $Z$- and $X$-bases. This corresponds to choosing the key map as
\begin{align}
\text{Key-map POVM:  }&Z_A = \{\dya{0}+\dya{2},\dya{1}+\dya{3}\}\,, 
\end{align}
where Alice's source-replacement state from Eq.~(12) is 
\begin{align}
\label{eqnsourcereplacemdi}
\ket{\psi_{AA'}}&= \sqrt{p_z /2}\left(\ket{0}\ket{0}+\ket{1}\ket{1}\right)+\sqrt{(1-p_z) /2}\left(\ket{2}\ket{+}+\ket{3}\ket{-}\right).
\end{align}
To obtain large key rates we employ biased basis choices \cite{Lo2004}, i.e., $p_z = 1-\epsilon$ with $0<\epsilon \ll 1$.  As noted above, we impose the correlation constraints in Supplementary Eq.~\eqref{eqnmdiapp6} as well as constraints that fix the form of the marginals $\rho_{A}$ and $\rho_B$, Supplementary Eq.~\eqref{eqnmdiapp7}. It is encouraging that our numerics reproduce the known theoretical curve \cite{Lo2012}, as shown in Fig.~2.

\section*{\uppercase{Supplementary Note 2: Arbitrary post-selection}}\label{apparbitrarypostselection}

In the Results section we discussed a method for transforming the constraints on $\rho_{AB}$ to constraints on the post-selected state $\GC(\rho_{AB})$ for the special case where $\GC$ has a CP inverse. We now generalize the method to transform the constraints for any CP map $\GC$.

The idea is to view the space of Hermitian operators as a vector space, and to partition the space into basis vectors whose coefficients are fixed by the constraints and those whose coefficients are free. Namely, we apply this view to the image space under post-selection, as follows.

By applying the Gram-Schmidt process to the measurement operators $\{\Gamma_i\}$, we can write the constraints on $\rho_{AB}$ equivalently as $\text{Tr}(\rho_{AB} \Delta_i) = \delta_i$ where the $\{\Delta_i\}$ are orthonormal under the Frobenius inner product $\langle A,B\rangle = \text{Tr}(A^\dagger B)$. We extend this to an orthonormal basis $\{\Delta_i\}\cup \{\Xi_j\}$ of the Hermitian operator space. Note that here, and in what follows, we take the basis elements to be Hermitian.

For clarity, in what follows we refer to the observation-based constraints on $\rho_{AB}$, which are generally of the form $\Tr (\rho_{AB} \Gamma_i)= \gamma_i$, as trace constraints. This is to distinguish them from the constraints on $\rho_{AB}$ due to its positivity. Now define $(\vec \delta)_i = \delta_i$ and $(\vec \Delta)_i = \Delta_i$, and similarly for $\vec \xi$ and $\vec \Xi$.  Then any $\rho_{AB}$ satisfying the trace constraints is of the form:
\begin{align}
\rho_{AB} = \vec \delta \cdot \vec \Delta + \vec \xi \cdot \vec \Xi = \rho_0 + \vec \xi \cdot \vec \Xi, 
\end{align}
for any $\vec \xi$. Note that the requirement that $\rho_{AB}\geq 0$ will constrain the possible values of $\vec \xi$, but for now we only consider the trace constraints. For later convenience, we have defined the (not necessarily positive semidefinite) operator $\rho_0 := \vec \delta\cdot \vec \Delta$.

Acting linearly with $\GC$, we find that the image of the set of $\rho_{AB}$ satisfying the trace constraints is
\begin{align}
\GC(\rho_{AB}) = \GC(\rho_0) + \vec \xi \cdot \GC(\vec \Xi).
\end{align}
Now, let $\{\Upsilon_m\}$ be an orthonormal basis for the space spanned by the operators $\{\GC(\Xi_j)\}$, which can again be found by the Gram-Schmidt process. We can extend this to an orthonormal basis $\{\Upsilon_m\}\cup \{\Omega_n\}$ of the image of the original Hilbert space under $\GC$, i.e., $\GC (\HC_{AB})$. In practice, this can be done by performing Gram-Schmidt again on $\{\Upsilon_m\}\cup \{\GC(\Delta_i)\}$, which will leave the $\Upsilon_m$ unchanged. The $\{\Omega_n\}$ are the operators of interest. We can calculate the coefficients $\omega_n$:
\begin{align}
\text{Tr}(\GC(\rho_{0})\Omega_n) = \omega_n.
\label{newconstraints}
\end{align}
With the basis decomposition $\GC(\rho_{AB}) = \vec \omega\cdot \vec \Omega + \vec \upsilon \cdot \vec \Upsilon$ for some $\vec \upsilon$, it can be shown that the trace constraints on $\rho_{AB}$ do not constrain $\vec \upsilon$ whatsoever. On the other hand, the coefficients in $\vec \omega$ are exactly determined in Supplementary Eq.~\eqref{newconstraints}. Thus, the trace constraints are exactly converted from $\rho_{AB}$ to $\GC(\rho_{AB})$ according to:
\begin{align}
\{\text{Tr}(\rho_{AB} \Gamma_i) = \gamma_i\}\Rightarrow \{\text{Tr}(\GC(\rho_{AB})\Omega_n) = \omega_n\}.\label{constraintmap}
\end{align}

Let us remark that the Hermitian operators in $\GC( \HC_{AB})$ will be represented as matrices in a possibly larger space $\tilde{\HC}_{AB} $, where $\GC( \HC_{AB}) \subseteq \tilde{\HC}_{AB} $. One has the freedom to choose $\tilde{\HC}_{AB} $ for a convenient matrix representation of $\GC (\HC_{AB})$, and hence $\tilde{\HC}_{AB} $ is not unique. But if $\tilde{\HC}_{AB} $ is strictly larger, i.e., $\GC( \HC_{AB}) \subset \tilde{\HC}_{AB} $, then we must enforce additional trace constraints, to restrict the optimization to $\mathcal G(\mathcal H_{AB})$. To obtain these additional constraints, complete the orthonormal basis $\{\Upsilon_m\}\cup\{\Omega_n\}$ of $\mathcal G(\mathcal H_{AB})$ to a basis of $\tilde{\HC}_{AB} $ with the additional orthonormal Hermitian operators $\{\Lambda_\ell\}$. The operators $O \in \tilde{\HC}_{AB} $ such that $O\in\mathcal G(\mathcal H_{AB})$ are exactly those that satisfy $\text{Tr}(O \Lambda_\ell) = 0$ for each $\ell$. Hence, one can add the constraints $\{\Tr (\GC(\rho_{AB})\Lambda_\ell ) = 0\}$ to the set in Supplementary Eq.~\eqref{constraintmap}. For future convenience, let us define the set of positive operators in $\tilde{\HC}_{AB} $ as $\tilde{\PC}_{AB}  := \{\tilde{\rho}_{AB} \in \tilde{\HC}_{AB} : \tilde{\rho}_{AB} \geq 0\}$.

With post-selection, the key rate formula in Eq.~$\devwin$ is applied to the post-selected state $\GC(\rho_{AB})/p_\text{pass}$, with the measured quantity $p_\text{pass} = \text{Tr}(\GC(\rho_{AB}))$. (See the remark in the main text where we note that taking Eve's system to purify the post-selected state does not introduce any looseness into our key rate calculation.) In the usual primal optimization, the optimization is taken over $\rho_{AB}$. However, we can directly reformulate it as an optimization over $\tilde \rho_{AB}$ in $\GC(\mathcal P_{AB})$, the image of $\PC_{AB}$ under the post-selection map. Let $\tilde{\mathcal B}$ be the set of $\tilde\rho_{AB}\in \GC(\mathcal P_{AB})$ satisfying the trace constraints $\text{Tr}(\tilde\rho_{AB} \Omega_n) = \omega_n$ for each $n$. Then the primal problem is:
\begin{align}
\widehat\alpha &= \min_{\rho_{AB}\in\mathcal C}\widehat \Phi\left(\frac{\GC(\rho_{AB})}{p_\text{pass}}, Z_A\right) \\
&=  \min_{\rho_{AB}\in\mathcal C} \min_{\sg_{ A B}\in\DC} \widehat{D}\bigg(\frac{\GC(\rho_{A B})}{p_\text{pass}}  \bigg|  \bigg| \ZC_A \big( \sigma_{AB} \big)\bigg) \\
&= \frac{\widetilde{\alpha}}{p_\text{pass}} - \ln p_\text{pass}\,,
\end{align}
where
\begin{align}
\widetilde{\alpha}:= \min_{\tilde\rho_{AB}\in\tilde{\mathcal B}}     \min_{\sg_{ A B}\in\DC} \widehat{D}\big(\tilde{\rho}_{A B} || \ZC_A (\sg_{A B} )\big) \,.
\end{align}

To apply our standard optimization algorithm, we need to optimize over a set of all positive semidefinite operators in a Hilbert space. Since $\GC$ is a CP map, $\GC(\mathcal P_{AB}) \subseteq \GC(\mathcal H_{AB})_+$, where $\GC(\mathcal H_{AB})_+$ is the set of positive semidefinite operators in $\GC(\mathcal H_{AB})$. The inclusion need not be with equality, so we have the inequality:
\begin{align}
\widetilde{\alpha} \ge  \min_{\tilde\rho_{AB}\in\tilde{\mathcal C}}   \min_{\sg_{ A B}\in\DC} \widehat{D}\big(\tilde{\rho}_{A B} || \ZC_A (\sg_{A B} )\big) ,\label{opreform}
\end{align}
where $\tilde{\mathcal C}$ is the set of $\tilde\rho_{AB}\in \GC(\mathcal H_{AB})_+$ such that $\text{Tr}(\tilde\rho_{AB} \Omega_n) = \omega_n$ for each $n$, or equivalently, the set of $\tilde\rho_{AB}\in \tilde{\PC}_{AB}$ such that $\text{Tr}(\tilde\rho_{AB} \Omega_n) = \omega_n$ for each $n$ and $\text{Tr}(\tilde\rho_{AB} \Lambda_\ell) = 0$ for each $\ell$. With the reformulation of the optimization problem in Supplementary Eq.~\eqref{opreform}, we have (at the expense of introducing an inequality) recast the optimization with post-selection into the usual form treated in the Methods section.

However, note that when $\GC$ has an inverse $\GC^{-1}$ that is CP, Supplementary Eq.~\eqref{opreform} is satisfied with equality. This follows from the fact that $\GC(\mathcal P_{AB}) = \GC(\mathcal H_{AB})_+$ in this case. This special case was discussed in the Results section. Furthermore, we note that the B92 protocol (see Supplementary Note~4) involves a post-selection map that has a CP inverse. So for that protocol, the step in Supplementary Eq.~\eqref{opreform} does not introduce any looseness.

\section*{\uppercase{Supplementary Note 3: Tightness for protocols with MUBs}}\label{apptightness}

Here we analytically prove Prop.~2. This states that our numerical approach is perfectly tight for the entanglement-based protocols involving MUBs discussed in the main text. 

First we note that the only potential source of looseness in our bound is our usage in Eq.~$\gtapplied$ of the Golden-Thompson (GT) inequality Eq.~$\gt$.   The question, then, is under what conditions is Eq.~$\gtapplied$  saturated.

\subsection*{A general lemma}

We begin by stating a general lemma, which gives a sufficient set of criteria that guarantee our method is tight. Note that these sufficient criteria might not be necessary for tightness.

\begin{lemma}\label{gtsatcon}
The GT inequality invoked in Eq.~$\gtapplied$ is saturated, and hence our method tight, for a QKD protocol satisfying the following two conditions:
\begin{enumerate}[label=(\alph*)]
\item $[\Gamma_{i},\Gamma_{i'} ]=0\,\,\, \forall\, i,i'$
\item $\bra{e_{\ell} } \ZC_A (\ket{e_k}\bra{e_k}) \ket{ e_{\ell '}}=0$ for $\ell \neq\ell '$ and $\forall k$ in a common eigenbasis $\{\ket{e_k}\}$ of all $\{\Gamma_i\}$.
\end{enumerate}
\end{lemma}

\begin{proof}
In general, the GT inequality Eq.~$\gt$ is satisfied with equality if and only if the two operators commute. In our case, the saturation of the GT inequality is equivalent to the vanishing of the following commutator
\begin{align}\left[Q(\lmv),\ln \ZC_A(\sigma_{AB}^*)\right]\,,
\label{commutatorGT1}\end{align}
where $Q(\lmv):= -\id - \lmv\cdot \Gmv$, and where $\sigma_{AB}^*$ is a maximal eigenvector of
\begin{align}
T:=\ZC_A \big(\exp( Q(\lmv) ) \big) = \ZC_A \big(R(\lmv) \big)\,,
\end{align}
i.e., an eigenvector of $T$ whose eigenvalue is the largest. In general $\sigma_{AB}^*$ is not uniquely defined if the maximal eigenvalue is degenerate. However, this issue does not affect the proof below. This is because, for tightness, we only need the GT inequality to be saturated for one particular $\sigma_{AB}^*$, i.e., one particular $\sigma_{AB}$ that achieves the optimization in Eq.~$\sigmaoptimal$.

As $\ZC_A(\sigma_{AB}^*)$ is positive semidefinite, it can be shown that the vanishing of Supplementary Eq.~\eqref{commutatorGT1}, and thus the saturation of the GT inequality, is equivalent to the vanishing of
\begin{align}\left[Q(\lmv),\ZC_A(\sigma_{AB}^*)\right].\label{commutatorGT2}\end{align}
This follows from the fact that $\ln \ZC_A(\sigma_{AB}^*)$ and $\ZC_A(\sigma_{AB}^*) = \exp (\ln \ZC_A(\sigma_{AB}^*))$ are diagonal in the same basis.

Now suppose that conditions (a) and (b) are satisfied. It follows from (a) that the measurement operators $\{\Gamma_i\}$ can be simultaneously diagonalized in an orthonormal eigenbasis $\{\ket{e_k}\}$. The operators $Q(\lmv)$ and $R(\lmv) = \exp(Q(\lmv))$ are also diagonal in such a basis.

From condition (b), we note that $\ZC_A$ maps an eigenstate $\ket{e_k}\bra{e_k}$ to a linear combination of $\dya{e_{\ell}}$ terms. Let the coefficients of that combination be $b_{k \ell}$ and let the eigenvalues of $R(\lmv)$ be $a_k$. Then $T$ is also diagonalizable in the $\{\ket{e_k}\}$ eigenbasis since

\begin{align}
T &=\sum_k a_k \ZC_A(\ket{e_k}\bra{e_k}) =  \sum_{\ell} \left(\sum_k b_{k \ell} a_k \right)\dya{e_{\ell}}.
\end{align}

Since $\sigma_{AB}^*$ is a maximal eigenvector of $T$, and $T$ is diagonal in the $\{\ket{e_k}\}$ basis, then let us choose $\sigma_{AB}^* = \dya{e_m}$ to correspond to a state $\ket{e_m}$ from this basis. While $T$ may have more than one eigenbasis, we remark that we have the freedom to choose $\sigma_{AB}^*$ from the $\{\ket{e_k}\}$ basis, since (as noted above) we only need the GT inequality to be saturated for a particular choice of $\sigma_{AB}^*$.

We find that Supplementary Eq.~\eqref{commutatorGT2} vanishes:
\begin{align}
\left[Q(\vec \lambda),\ZC_A(\sigma_{AB}^*)\right] & = \left[Q(\vec\lambda),\sum_{\ell}b_{m \ell}\dya{e_{\ell}}\right] =\sum_{\ell} b_{m \ell}\left[Q(\lmv),  \dya{e_{\ell}}   \right] = 0\,,
\end{align}
and thus the GT inequality is saturated if conditions (a) and (b) are satisfied.
\end{proof}

\subsection*{Specific protocols}

We now show that conditions (a) and (b) in Supplementary Lemma~\ref{gtsatcon} are satisfied for the protocols involving MUBs in the main text.

First we define some notation. The generalized Pauli operators in dimension $d$ are
\begin{align}
\label{eqnpauliopsz}
 \sigma_Z &:= \sum_j \omega^{j} \dya{j}\\
\label{eqnpauliopsx}
 \sigma_X &:= \sum_j \dyad{j+1}{j} = F \sigma_Z F\ad\,,
\end{align}
with $\omega = e^{2\pi i/ d}$. From these operators one can construct the Bell basis states $\{\ket{\phi_{q,r}}\}$, i.e., a set of $d^2$ orthonormal states of the form
\begin{align}
\label{eqnbellbasis}
\ket{\phi_{q,r}} := \id \ot \sigma_X^q \sigma_Z^r \ket{\phi_{0,0}},\quad\text{with }q,r \in \{0,...,d-1\}\,,
\end{align}
where 
\begin{align}
\label{eqnbellbasis2}
\ket{\phi_{0,0}} := \sum_j \frac{1}{\sqrt{d}}\ket{j}\ket{j}\,.
\end{align}

Our proof of tightness will proceed by showing that the $\Gamma_i$ operators of interest are all diagonal in the Bell basis (Supplementary Eq.~\eqref{eqnbellbasis}), and furthermore that the Bell basis satisfies condition (b) in Supplementary Lemma~\ref{gtsatcon}. Let us first show the latter, since it will be used repeatedly below.

\begin{lemma}\label{propbellbasis}
The Bell basis $\{\ket{\phi_{q,r}}\}$ satisfies condition (b) in Supplementary Lemma~\ref{gtsatcon}. That is, 
\begin{align}
\label{eqnbellbasis346}
\mted{\phi_{\ell,m}}{\ZC_A(\dya{\phi_{q,r}})}{\phi_{\ell ' , m'}} = 0,\quad  \forall (\ell,m) \neq (\ell ' , m')\text{ and } \forall (q,r)\,.
\end{align}
\end{lemma}
\begin{proof}
In all the protocols under consideration,
\begin{align}
\label{eqnbellbasis347}
Z_A = \{\dya{j}\}_{j=0}^{d-1}
\end{align}
is taken to be the standard basis on system $A$. Hence we can rewrite the action of the channel $\ZC_A$ on some operator $O$ as
\begin{align}
\label{eqnbellbasis348}
\ZC_A(O) &= \sum_j \dya{j} O \dya{j}\\
 &= \frac{1}{d}\sum_{j,k,k'} \omega^{j(k-k')}\dya{k} O \dya{k'}\\
 &= \frac{1}{d}\sum_{j} \sigma_Z^j  O (\sigma_Z^j)\ad\,.
\end{align}
Next note that $\ket{\phi_{0,0}}$ has the property
\begin{align}
\label{eqnbellbasis350}
(O\ot \id ) \ket{\phi_{0,0}} = (\id \ot O^T ) \ket{\phi_{0,0}}
\end{align}
for some operator $O$, where $^T$ is the transpose in the standard basis. Hence we have
\begin{align}
\label{eqnbellbasis348}
\ZC_A(\dya{\phi_{0,0}}) &= \frac{1}{d}\sum_{j} (\sigma_Z^j \ot \id) \dya{\phi_{0,0}} ((\sigma_Z^j)\ad \ot \id)\\
\label{eqnbellbasis348b}
&= \frac{1}{d}\sum_{j} (\id \ot \sigma_Z^j ) \dya{\phi_{0,0}} (\id \ot (\sigma_Z^j)\ad )   \,.
\end{align}
Finally, using the definition in Supplementary Eq.~\eqref{eqnbellbasis}, we have 
\begin{align}
\label{eqnbellbasis348}
\ZC_A(\dya{\phi_{q,r}}) &= (\id \ot \sigma_X^q \sigma_Z^r) \ZC_A(\dya{\phi_{0,0}})(\id \ot \sigma_X^q \sigma_Z^r)\ad\\
&= \frac{1}{d}\sum_{j} (\id \ot \sigma_X^q \sigma_Z^{r+j }) \dya{\phi_{0,0}} (\id \ot \sigma_X^q \sigma_Z^{r+j })\ad   \\
\label{eqnbellbasis349}&= \frac{1}{d}\sum_{j} \dya{\phi_{q,r+j}}  \,.
\end{align}
Clearly Supplementary Eq.~\eqref{eqnbellbasis349} is diagonal in the Bell basis, proving the desired result.
\end{proof}
Therefore, in the specific protocols considered below, we only need to show that the $\{\Gamma_i\}$ operators are diagonal in the Bell basis, to prove tightness of our method.

\subsubsection*{Two MUBs}

First let us consider the protocol discussed in the main text involving only two MUBs in arbitrary dimension $d$. Here the $\{\Gamma_i\}$ operators are $\{\id, E_Z, E_X\}$, where we write the error operators as $E_Z = \id - C_Z$ and $E_X = \id - C_X$, with
\begin{align}
\label{eqnEZdefarbd}
C_Z &:=  \sum_j \dya{j} \ot \dya{j} \\
\label{eqnEXdefarbd}
C_X &:=  \sum_j F\dya{j} F\ad \ot F\ad\dya{j} F.
\end{align}

It suffices to show that $C_Z$ and $C_X$ are diagonal in the Bell basis. First, note that $C_Z$ is $d$ times the quantity in Supplementary Eq.~\eqref{eqnbellbasis348b}, and hence
\begin{align}
\label{eqnCZ}
C_Z &=  \sum_{r=0}^{d-1} \dya{\phi_{0,r}},
\end{align}
which is obviously diagonal in the Bell basis. Next we write
\begin{align}
\label{eqnCZ}
C_X &=  (F \ot F\ad)C_Z(F\ad \ot F)\\
 &=    \sum_r (F \ot F\ad) \dya{\phi_{0,r}} (F\ad \ot F)\\
 &=    \sum_r (F \ot F\ad \sg_Z^r) \dya{\phi_{0,0}} (F\ad \ot (\sg_Z^r)\ad F )\\
 &=    \sum_r (\id \ot F\ad \sg_Z^r F) \dya{\phi_{0,0}} (\id \ot F\ad(\sg_Z^r)\ad F )\\
\label{eqnCZ4}
 &=    \sum_r (\id \ot \sg_X^{d-r}) \dya{\phi_{0,0}} (\id \ot ( \sg_X^{d-r})\ad )\\
 &=    \sum_r  \dya{\phi_{r,0}} \,,
\end{align}
where Supplementary Eq.~\eqref{eqnCZ4} used the relation
\begin{align}
\label{eqnCZ8}
F\ad \sg_Z^r F= (F \sg_Z^r F\ad )^T= (\sg_X^{r})^T =  \sg_X^{d-r} \,.
\end{align}
Clearly the final expression for $C_X$ is diagonal in the Bell basis. This proves that our numerical approach is tight for the protocol discussed in Fig.~3 of the main text.

\subsubsection*{Six-state protocol}

The six-state protocol (see Fig.~1) is a qubit protocol involving the operators $\{\Gamma_i\}=\{\id, E_Z, E_X+E_Y\}$. We already showed that $E_Z$ and $E_X$ are diagonal in the Bell basis, so we just need to do the same for $E_Y$. Note that we can write $E_Y = (\id - \sigma_Y \ot \sigma_Y)/2$, where $\sigma_Y := -i\dyad{0}{1}+i\dyad{1}{0}$. So it suffices to show that $\sigma_Y \ot \sigma_Y$ is diagonal in the Bell basis. This follows from directly computing the action on the four Bell states:
\begin{align}
\label{eqnsixstatebell}
\sigma_Y \ot \sigma_Y ( \ket{00}\pm \ket{11}) &= -( \ket{11}\pm \ket{00})\\
\sigma_Y \ot \sigma_Y ( \ket{01}\pm \ket{10}) &= ( \ket{10}\pm \ket{01})\,.
\end{align}
Hence our method is tight for the six-state protocol.

\subsubsection*{$n$ MUBs}

The protocol considered in Fig.~4 involved $n$ MUBs in $d=5$. These MUBs were chosen based on a construction in Ref.~\cite{Bandyopadhyay2001}. Namely, in prime dimension, the eigenvectors of the operators
\begin{align}
\label{eqnmubprimed}
\sigma_Z, \sigma_X, \sigma_X \sigma_Z, ...,  \sigma_X \sigma_Z^{d-1}
\end{align}
form a set of $d+1$ MUBs. In Fig.~4, we considered a subset of size $n$ of the MUBs in Supplementary Eq.~\eqref{eqnmubprimed}.

The measurement operators are:
\begin{align}
\label{eqnmubprimed2}
\{\Gamma_i\}  = \{\id, E_Z, E_X + E_{XZ}+ ...  + E_{XZ^{n-2}} \}
\end{align}
where $E_{XZ^k}$ denotes the error operator for the basis associated with $\sigma_X \sigma_Z^{k}$.  

We already showed above that $E_Z$ and $E_X$ are diagonal in the Bell basis, so it remains to show this for $E_{XZ}, ... ,E_{XZ^{n-2}}$. Again let us use the notation
\begin{align}
\label{eqnmubprimed3}
C_{XZ^k}&:= \id -E_{XZ^k} \\
&= (H_k \ot H_k^*) C_Z (H_k \ot H_k^*)\ad \,,
\end{align}
where $H_k$ is the Hadamard (unitary) matrix that rotates the standard basis to the eigenbasis of $\sigma_X \sigma_Z^{k}$, and $H_k^*$ denotes its conjugate in the standard basis.

Consider the case where $d$ is an odd prime. Note that the only even prime is $d=2$ which we already covered above. We restrict to odd primes here, since the following construction applies to them
\begin{align}
\label{eqnmubprimed4}
H_k = \sum_{j,j'} \frac{1}{\sqrt{d}} \omega^{-jj' - k s_j}\dyad{j}{j'}
\end{align}
where $s_j := (d-j)(d+j-1)/2$.

Proceeding similarly to Supplementary Eq.~\eqref{eqnCZ}, we write
\begin{align}
\label{eqnCxzk}
C_{XZ^k} & =  \sum_r (H_k \ot H_k^*) \dya{\phi_{0,r}}  (H_k \ot H_k^*)\ad \\
 &=    \sum_r (H_k \ot H_k^* \sg_Z^r) \dya{\phi_{0,0}} (H_k\ad \ot (\sg_Z^r)\ad H_k^T )\\
 &=    \sum_r (\id \ot H_k^* \sg_Z^r H_k^T) \dya{\phi_{0,0}} (\id \ot H_k^* (\sg_Z^r)\ad H_k^T )\\
\label{eqnCxzk2}
 &=    \sum_r (\id \ot \sg_X^{d-r}\sg_Z^{kr}) \dya{\phi_{0,0}} (\id \ot ( \sg_X^{d-r} \sg_Z^{kr})\ad )\\
 &=    \sum_r  \dya{\phi_{d-r,kr}} \,,
\end{align}
which is diagonal in the Bell basis. In Supplementary Eq.~\eqref{eqnCxzk2}, we used
\begin{align}
\label{eqnmubprimed5}
H_k^* \sg_Z^r H_k^T = \omega^{-kr(r+1)/2} \sg_X^{d-r}\sg_Z^{kr},
\end{align}
and noted that the phase factor $\omega^{-kr(r+1)/2}$ disappears when multiplied by its conjugate.

\section*{\uppercase{Supplementary Note 4: Analysis of B92 protocol}}\label{appb92}

Here we elaborate on our analysis of the B92 protocol. Recall that Alice sends one of two non-orthogonal states $\{\ket{\phi_0} , \ket{\phi_1}\}$ to Bob, and Bob randomly measures either in basis $B_0 = \{\ket{\phi_0},\ket{\overline{\phi_0}}\}$ or basis $B_1 = \{\ket{\phi_1},\ket{\overline{\phi_1}}\}$, where $\ip{\phi_0}{\overline{\phi_0}}=\ip{\phi_1}{\overline{\phi_1}}=0$. They post-select on rounds where Bob gets outcome $\ket{\overline{\phi_0}}$ or $\ket{\overline{\phi_1}}$.

Since this is a prepare-and-measure protocol, we use the source-replacement scheme as outlined in the Results section. That is, Alice and Bob obtain constraints on the state $\rho_{AB}$ in Eq.~(13). The optimization problem is then defined as follows
\begin{align}
\label{eqnb92constraintsA}\text{Key-map POVM:  }&Z_A = \{\dya{0},\dya{1}\}  \\
\label{eqnb92constraintsB}\text{Constraints:  }&\ave{\id} = 1 \\
\label{eqnb92constraintsC}&\ave{\Gamma_1} = p/2\\
\label{eqnb92constraintsD}&\ave{\Gamma_2} =p/2+ (1-p)\sin^2(\theta/2)\\
\label{eqnb92constraintsE}&\ave{\sigma_X \ot \id} = \cos (\theta/2)\\
\label{eqnpostselectb92}\text{Post-selection:  }
&G = \id_A \ot \left[\frac{1}{2}\big(\dya{\overline{\phi_0}}+\dya{\overline{\phi_1}}\big)\right]^{1/2}\,.
\end{align}
Here, $p$ is the depolarizing probability, $\sigma_X = \dyad{0}{1}+\dyad{1}{0}$, and the error and success operators are respectively
\begin{align}
\Gamma_1 &:= \dya{0}\ot \dya{\overline{\phi_0}}+\dya{1}\ot \dya{\overline{\phi_1}}\\
\Gamma_2 &:= \dya{0}\ot \dya{\overline{\phi_1}}+\dya{1}\ot \dya{\overline{\phi_0}}\,.
\end{align}
Note that the constraint in Supplementary Eq.~\eqref{eqnb92constraintsE} serves to constrain $\rho_A$ and is of the form of Eq.~(15) (see discussion around Eq.~(14)). In principle one can add additional constraints on $\rho_A$, although we found this did not affect the key rate. In addition to the key map and constraints, note that we also needed to define the filter associated with Bob's post-selection; see Eq.~(18) in the Results section. Here we wrote the post-selection map in Kraus form: $\GC(O) = GOG\ad$ for any operator $O$. 

Supplementary Eqs.~\eqref{eqnb92constraintsA}-\eqref{eqnpostselectb92} define the optimization problem, and our results are shown in Fig.~6. We remark that our formulation of the dual problem has only 4 parameters, whereas the primal problem has 12 parameters, making the latter somewhat more difficult to solve.

\section*{\uppercase{Supplementary Note 5: Arbitrary key-map POVMs}}\label{apppovm}

The Methods section proves our main result for the case where the key-map POVM is a projective measurement. Here we generalize this proof to arbitrary measurements.

First we rewrite the primal problem in terms of a coherence-like quantity, similar to what we did in the Methods section. Consider some POVM $P=\{P_j\}$ with $P_j \geq 0$ for each $j$, and $\sum_j P_j = \id$. Let us first define a generalized notion of coherence as follows,
\begin{align}
\label{eqncoherencegeneralizeddef}
\Phi_G(\rho, P) := D\bigg(\rho  \bigg| \bigg| \sum_j P_j \rho P_j \bigg)\,.
\end{align}
We note that the second argument $\sum_j P_j \rho P_j$ is not necessarily normalized, although in general we have $\Tr(\sum_j P_j \rho P_j )\leq 1$, which follows from $P_j^2 \leq P_j$. In turn, this implies that $\Phi_G$ is non-negative: 
\begin{align}
\label{eqncoherencegeneralizeddef22}
\Phi_G(\rho, P) \geq 0\,.
\end{align}

Suppose Alice's measurement is an arbitrary POVM, $Z_A = \{Z_A^j\}$. Then, from Lemma 4 of Ref.~\cite{Coles2011}, we have:
\begin{align}
\label{eqncoherencegeneralizedrelation}
H(Z_A | E) \geq \Phi_G(\rho_{AB}, Z_A) 
\end{align}
where $E$ can be taken to purify $\rho_{AB}$. Hence we can define (or lower bound) the primal problem as
\begin{align}
\label{eqnpovmprimalproblem4}
\al := \min_{\rho_{AB}\in \CC} \Phi_G(\rho_{AB}, Z_A) \,,
\end{align}
which is analogous to Eq.~$\alphacoherence$.

One can then transform to the dual problem, as described in the Methods section. The only subtlety is that the analog of Eq.~$\coherencesigma$ can be written as an inequality: 
\begin{align}
\label{eqn2356796povm}
\widehat{\Phi}_G(\rho_{AB}, Z_A) &= \widehat{D}\big(\rho_{A B} || \ZC_A (\rho_{A B} ) \big) \\
\label{eqn23532625povm}
&\geq \min_{\sg_{ A B}\in\DC} \widehat{D}\big(\rho_{A B} || \ZC_A ( \sg_{A B} ) \big).
\end{align}
The rest of the derivation proceeds as described in the Methods section.

\section*{\uppercase{Supplementary Note 6: Strong duality}}\label{appstrongdual}

In general the dual problem gives a lower bound on the primal problem, a fact called weak duality \cite{Boyd2010}. In the notation used in the Methods section, this means that
\begin{align}
\label{eqnweakdual}
\widehat{\beta} \leq \widehat{\al} \,.
\end{align}
However, under certain conditions, the dual and primal problems are equivalent, which is called strong duality. Here we show that strong duality holds for our problem. That is, we prove Eq.~$\strongduality$, $\widehat{\beta} = \widehat{\al}$.

Slater's condition for a convex optimization problem is a sufficient criterion that guarantees strong duality \cite{Boyd2010}. For a problem with affine constraints, Slater's condition is satisfied when there is an element of the relative interior of the domain of optimization that satisfies the constraints. Thus, one way to guarantee strong duality for our problem is to show that there is a $\rho_{AB}$ satisfying the constraints in Eq.~$\constraints$ such that $\rho_{AB} > 0$. However, it is easy to imagine examples where Alice's and Bob's constraints specify a 0 eigenvalue for $\rho_{AB}$. For such examples, one would not satisfy Slater's condition, since there would be no $\rho_{AB} > 0$ consistent with the constraints.

We get around this issue as follows. We show that by slightly perturbing the constraints in a physical way, we can arrive at a ``perturbed'' problem which \emph{does} satisfy Slater's condition and thus strong duality. Intuitively, changing the constraints infinitesimally should not change the solution to the primal problem by more than an infinitesimal amount; we prove this formally using the continuity of the coherence (see Supplementary Lemma~\ref{lemcontinuity} below). Thus solving the dual optimization (i.e., solving for $\widehat{\beta}$) can be done for slightly perturbed constraints, giving a solution arbitrarily close, and hence equivalent for all practical purposes, to the solution of the primal problem, $\widehat{\al}$. It is in this sense that we have strong duality in Eq.~$\strongduality$.

(For simplicity, we drop the subscript $AB$ on $\rho_{AB}$ for the remainder of this section, and just write $\rho$. Also, we assume that the $\{Z_A^j\}$ are projectors, and we replace $Z_A = \{Z_A^j\}$ with a generic set of projectors $\Pi = \{\Pi^j\}$.)

More precisely, we consider the following two optimization problems
\begin{align}
&\text{Problem 1:  }a_1 = \min_{\rho \in \SC_1} \Phi(\rho, \Pi)\label{unpert}\\
&\text{Problem 2:  }a_2(\varepsilon ) = \min_{\rho \in \SC_2 (\varepsilon)} \Phi(\rho, \Pi)\label{pert}
\end{align}
where one should recall that the primal problem objective function was identified in Eq.~$\alphacoherence$ as the coherence $\Phi$. Here,
\begin{align}
\SC_1 &:= \{\rho \in \HC_d : \rho \geq 0, \Tr(\rho \vec{\Gamma}) =\gmv \}\notag\\
\SC_2 (\varepsilon) &:= \{\rho \in \HC_d : \rho  > 0, \Tr(\rho \vec{\Gamma}) = (1-d \varepsilon)\gmv+\varepsilon \Tr(\Gmv) \}.\notag
\end{align}
We call Supplementary Eq.~\eqref{unpert} the unperturbed problem and Supplementary Eq.~\eqref{pert} the perturbed problem. In Supplementary Propositions \ref{sdpert} and \ref{pertlimit}, we shall prove that strong duality holds for the perturbed problem and that
\begin{align}
\label{eqnlimit1111}
\lim_{\varepsilon\rightarrow 0^+}a_2(\varepsilon) = a_1.
\end{align}
By choosing $\varepsilon$ sufficiently small, solving the dual perturbed problem will then yield a result that is arbitrarily close to the solution to the primal problem.
\vspace{4pt}

\begin{proposition} The strong duality property holds for the perturbed problem.\label{sdpert}\end{proposition}
\begin{proof}
 We prove this by showing that Slater's condition is satisfied for the perturbed primal problem, which implies strong duality for the perturbed problem \cite{Boyd2010}. Slater's condition for this convex optimization problem with affine constraints is satisfied when there is a $\rho>0$ that satisfies the constraints.

We begin by considering a map from the domain $\SC_1$ to $\SC_2  (\varepsilon)$. Let the map $\MC_{\varepsilon}$ act on a state $\rho$ via
\begin{align}\MC_{\varepsilon}(\rho)= (1-d \varepsilon)\rho + \varepsilon \id.\label{injection}\end{align}
Note that if $\rho \in \SC_1$, then $\MC_{\varepsilon}(\rho) \in \SC_2  (\varepsilon)$.

We take the given set of constraints to be physical, i.e., $\SC_1$ is non-empty. By the map $\MC_\varepsilon$, $\SC_2(\varepsilon)$ is non-empty as well. Then, as $\rho > 0$ for all $\rho\in\SC_2(\varepsilon)$, Slater's condition is satisfied. 
\end{proof}

Before proving the continuity of the perturbation, we note that coherence $\Phi$ has some nice properties. One property that we will make explicit use of is its \textit{continuity} in the state $\rho$, which we prove in the following lemma. After completion of this work, Ref.~\cite{Winter2015} proved the same lemma in their article.
\begin{lemma}
\label{lemcontinuity}
Let $\rho$ and $\sg$ be two density operators on a Hilbert space of dimension $d$. Suppose they are close in trace distance $T(\tau, \tau ' ):= (1/2)\Tr |\tau  -  \tau ' |$, in particular, suppose $T(\rho, \sg ) \leq 1/e$. Then the coherences of $\rho$ and $\sg$ are nearly equal:
\begin{align}
\label{eqncoherencecontinuity}
\Delta \Phi &:= | \Phi(\rho, \Pi) - \Phi(\sg, \Pi) | \notag \\
&\leq  2 [ T(\rho, \sg ) \log_2 d - T(\rho, \sg ) \log_2  T(\rho, \sg ) ].
\end{align} 
\end{lemma}
\begin{proof}
The proof uses Fannes' inequality, which states that
\begin{align}
\label{eqnFannes}
| H(\rho ) - H(\sg) | \leq   T(\rho, \sg ) \log_2 d - T(\rho, \sg ) \log_2  T(\rho, \sg ) ,
\end{align} 
which holds so long as $T(\rho, \sg ) \leq 1/e$. Note that, because of the monotonicity of the trace distance under quantum channels, we also have $T(\rho_{\Pi}, \sg_{\Pi} ) \leq 1/e$, where $\rho_{\Pi} = \sum_j \Pi^j \rho \Pi^j$ and $\sg_{\Pi} = \sum_j \Pi^j \sg \Pi^j$. Hence Supplementary Eq.~\eqref{eqnFannes} also holds for the states $\rho_{\Pi}$ and $\sg_{\Pi}$.

Noting that $\log \rho_{\Pi} = \sum_j \Pi^j (\log \rho_{\Pi}) \Pi^j$, we have
$$H(\rho_{\Pi}) = - \Tr(\rho \log \rho_{\Pi}),$$
and hence we can rewrite the coherence as 
$$\Phi(\rho, \Pi) = H(\rho_{\Pi}) - H(\rho).$$
This allows us to bound the coherence difference:
\begin{align}
\label{eqnFannesCohbound}
\Delta \Phi &=   | H(\rho_{\Pi}) - H(\sg_{\Pi}) + H(\sg) - H(\rho) | \\
&\leq | H(\rho_{\Pi}) - H(\sg_{\Pi}) | + | H(\sg) - H(\rho) | \\
&\leq T(\rho_{\Pi} , \sg_{\Pi} ) \log_2 d - T(\rho_{\Pi} , \sg_{\Pi} ) \log_2  T(\rho_{\Pi}, \sg_{\Pi} ) +T(\rho, \sg ) \log_2 d - T(\rho, \sg ) \log_2  T(\rho, \sg ) \\
&\leq 2 [T(\rho, \sg ) \log_2 d - T(\rho, \sg ) \log_2  T(\rho, \sg )],
\end{align} 
where the last line uses $T(\rho_{\Pi}, \sg_{\Pi} )\leq T(\rho , \sg )$, as well as the monotonicity of $(-x\log x)$ over the interval $x\in [0,1/e]$.
\end{proof}

To prove Supplementary Eq.~\eqref{eqnlimit1111}, we first state the following two technical lemmas.
\begin{lemma}
\label{lemcontinuity223}
There exists a $\rho \in\SC_2(\varepsilon)$ such that $\Phi(\rho ,\Pi)$ is within $\mathcal O(\varepsilon)$ of the unperturbed solution, $a_1$.
\end{lemma}
\begin{proof}
Let $\overline{\rho} \in \SC_1$ be a positive semidefinite matrix that minimizes the unperturbed problem, i.e.,
\begin{align}
\label{eqnrhobar1}
a_1 = \Phi(\overline{\rho} ,\Pi)\,.
\end{align} 
Consider the corresponding positive definite matrix $\mathcal M_\varepsilon(\overline{\rho} )\in\SC_2(\varepsilon)$, where $\mathcal M_\varepsilon:\SC_1 \rightarrow \SC_2(\varepsilon)$ is the injection defined in Supplementary Eq.~\eqref{injection}. We argue that this injection does not change the objective function much, due to continuity of coherence. Note that the states $\overline{\rho} $ and $\MC_{\varepsilon}(\overline{\rho} )$ are close in trace distance: 
\begin{align}
\label{eqnrhobar2}
T(\overline{\rho} , \MC_{\varepsilon}(\overline{\rho} )) = d \varepsilon T(\overline{\rho}  , \id / d) \leq d \varepsilon \,.
\end{align} 
From Supplementary Eq.~\eqref{eqncoherencecontinuity} we have that
\begin{align}
\label{eqncoherencebijection}
| \Phi( \overline{\rho} , \Pi) - \Phi(\MC_{\varepsilon}(\overline{\rho} ), \Pi) | \leq  -2  d \varepsilon \log_2  \varepsilon \,,
\end{align} 
which proves the statement.
\end{proof}

\begin{lemma}
\label{lemcontinuity222}
There exists a $\rho\in\SC_1$ such that $\Phi(\rho,\Pi)$ is within $\mathcal O(\varepsilon)$ of the perturbed solution, $a_2(\varepsilon)$.
\end{lemma}
\begin{proof}As in Supplementary Note 2, we consider the trace constraints rephrased (via the Gram-Schmidt process) in terms of an orthonormal set of operators $\{\Delta_i\}$, i.e.,
\begin{align}
\label{eqnonbasisoperators1}
\text{Tr}(\rho\vec\Gamma) = \vec\gamma \quad \to \quad \text{Tr}(\rho\vec\Delta) = \vec\delta\,,
\end{align}
and we extend the set to an orthonormal basis of Hermitian operators:
\begin{align}
\label{eqnonbasisoperators2}
\{\Delta_i\}\cup\{\Xi_j\}\,.
\end{align}

For $\rho_1\in\SC_1$ and $\rho_2\in\SC_2(\varepsilon)$, the most general expressions for these states are:
\begin{align}
\label{eqnrho1gen}\rho_1 &= \vec\delta\cdot\vec\Delta + \vec \xi_1 \cdot \vec \Xi\\
\label{eqnrho2gen}\rho_2 &= (\vec\delta + \vec p)\cdot\vec\Delta + \vec \xi_2\cdot \vec \Xi,\end{align}
where $\vec\xi_1$ is any vector that enforces the positive semidefiniteness of $\rho_1$ and $\vec\xi_2$ is any vector that enforces the positive definiteness of $\rho_2$. In Supplementary Eq.~\eqref{eqnrho2gen}, we used the fact that the constraint perturbation from the unperturbed problem in Supplementary Eq.~\eqref{unpert} to the perturbed problem in Supplementary Eq.~\eqref{pert} takes the form $\vec \delta \rightarrow \vec \delta + \vec p$, with $\vec p := \varepsilon\vec p_0$, and
\begin{align}
\vec p_0 := \text{Tr}(\vec\Delta) - d\hspace{1.6 pt}\vec\delta.
\end{align}

In what follows, we study the geometry of the set of possible $\vec \xi_1$ and $\vec \xi_2$, thereby finding the most general form for $\rho_1$ and $\rho_2$.

For all states $\ket i$ in the Hilbert space $\mathcal H$, the positive (semi)definiteness requirements are equivalent to the following statements, for all $\ket i$
\begin{align}\label{conon1}\bra{i}\rho_1\ket i &= \vec\delta \cdot \bra i \vec\Delta\ket i + \vec\xi_1 \cdot \bra i \vec\Xi\ket i \ge 0 \\
\bra{i}\rho_2\ket i &= (\vec\delta + \vec p) \cdot \bra i \vec\Delta\ket i + \vec\xi_2 \cdot \bra i \vec\Xi\ket i > 0. \label{conon2}\end{align}

Define the infinite-dimensional operators $C$ and $D$ with $(C)_{ij} := \bra i (\vec\Xi)_j\ket i$ and $(D)_{ij} := \bra i (\vec \Delta)_j \ket i$. Then the $\vec \xi$ which enforce the positive (semi)definiteness requirements are those which satisfy the following infinite vector of inequalities
\begin{align} C\cdot \vec\xi_1 &\ge - D\cdot \vec\delta = \vec s  \\
 C\cdot  \vec\xi_2 &> -D\cdot \vec \delta - \varepsilon D \cdot \vec p_0 = \vec s + \varepsilon \vec t \,, \end{align}
where the right sides are fixed quantities, so we defined $\vec s := - D \cdot \vec \delta$ and $\vec t :=- D\cdot \vec p_0$.

For each $\ket i$, the solution space of $\vec\xi_1$ and $\vec \xi_2$ satisfying the inequality constraint is a hyperplane defined by
\begin{align}&C_{i1}\big(\vec{\xi}_1\big)_1 + C_{i2}\big(\vec{\xi}_1 \big)_2 + \cdots \ge s_i\\
&C_{i1}\big(\vec{\xi}_2 \big)_1 + C_{i2}\big(\vec{\xi}_2 \big)_2 + \cdots > s_i + \varepsilon t_i.\end{align}
The regions of all possible $\vec \xi_1$ and $\vec \xi_2$ are the intersections of all of the constrained regions, each bounded by a hyperplane. Since $\SC_1$ and $\SC_2(\varepsilon)$ are non-empty, the solution sets for $\vec\xi_1$ and $\vec\xi_2$ are non-empty. Perturbing the constraints is equivalent to shifting each hyperplane slightly (by $\varepsilon t_i$ for the $i$th hyperplane). For sufficiently small $\varepsilon$, each $\vec\xi_2$ in the perturbed solution set is within $\mathcal O(\varepsilon)$ in Euclidean distance of a $\vec\xi_1$ in the unperturbed solution set \cite{Hu1989}.

Now, let $\rholh_2 \in \SC_2(\varepsilon)$ be a positive definite matrix that minimizes the perturbed problem, i.e., 
\begin{align}
\label{eqnrho1hat0a}
a_2(\varepsilon) = \Phi(\rholh_2 ,\Pi)\,.
\end{align}
In terms of the basis in Supplementary Eq.~\eqref{eqnonbasisoperators2},
\begin{align}
\label{eqnrho2hat}
\rholh_2 = (\vec \delta + \varepsilon\vec p_0)\cdot\vec\Delta + \vec{\hat{\xi}}_2 \cdot\vec\Xi \,,
\end{align}
for a vector $\vec{\hat{\xi}}_2$ that enforces positive definiteness for $\rholh_2$.

By the previous discussion, there exist a $\vec{\hat{\xi}}_1$ which enforces positive semidefiniteness for a matrix satisfying the unperturbed trace constraints that is within $\mathcal O(\varepsilon)$ of $\vec{\hat{\xi}}_2$. That is, the vector
\begin{align}
\label{eqnrho1hat0}
\vec{v} : = \vec{\hat{\xi}}_2 - \vec{\hat{\xi}}_1
\end{align}
has a norm $| \vec{v} |$ that is $\mathcal O(\varepsilon)$, and hence each component of $\vec{v}$ is no larger in magnitude than $\mathcal O(\varepsilon)$. We define
\begin{align}
\label{eqnrho1hat}
\rholh_1 := \vec\delta \cdot \vec \Delta + \vec{\hat{\xi}}_1\cdot\vec \Xi\in\SC_1.
\end{align}
The difference between Supplementary Eqs.~\eqref{eqnrho2hat} and \eqref{eqnrho1hat}
\begin{align}
\label{eqnrhohatdiff}
\rholh_2 - \rholh_1 = \varepsilon \vec p_0 \cdot \vec \Delta + \vec{v} \cdot\vec \Xi
\end{align}
has trace norm $\mathcal O(\varepsilon)$, as it is a sum of Hermitian operators with coefficients that are $\mathcal O(\varepsilon)$. Hence we have constructed an element $\rholh_1\in\SC_1$ that is within $\mathcal O(\varepsilon)$ in trace distance of $\rholh_2$. Then by continuity of coherence in Supplementary Lemma \ref{lemcontinuity}, the objective function value $\Phi(\rholh_1,\Pi)$ is within $\mathcal O(\varepsilon)$ of $a_2(\varepsilon)$.
\end{proof}

\begin{proposition} \label{pertlimit}For $a_1,a_2(\varepsilon)$ defined in Supplementary Eqs.~\eqref{unpert} and \eqref{pert},
\begin{align}\lim_{\varepsilon\rightarrow 0^+}a_2(\varepsilon) = a_1.
\end{align} \end{proposition}
\begin{proof}
Supplementary Lemma~\ref{lemcontinuity223} implies that
\begin{align}
\lim_{\varepsilon\rightarrow 0^+}a_2(\varepsilon) \leq a_1\,.
\end{align}
Supplementary Lemma~\ref{lemcontinuity222} implies that
\begin{align}
\lim_{\varepsilon\rightarrow 0^+}a_2(\varepsilon) \geq a_1\,.
\end{align}
Combining these two facts gives the desired result.
\end{proof}
\vspace{3pt}

\end{document}